\documentclass[12pt]{article}
\usepackage{graphicx}
\usepackage{mathrsfs}
\usepackage{fixltx2e}
\usepackage{marvosym}
\usepackage{amsmath}
\usepackage{amssymb}
\usepackage{amsbsy}
\usepackage[T1]{fontenc}
\usepackage{amsthm}
\usepackage{manfnt}
\usepackage{algorithm}
\usepackage{algpseudocode}
\usepackage{caption}
\usepackage{datetime}
\usepackage{mathtools}
\usepackage{stmaryrd}
\usepackage{authblk}
\usepackage{tikz}
\usepackage{fullpage}
\usepackage{nicefrac}
\usepackage[hyphens]{url}
\usepackage{listings}

\usepackage{hyperref}

\DeclareGraphicsExtensions{.png, .jpg}

\theoremstyle{plain}
\newtheorem{lem}{Lemma}
\newtheorem{thm}{Theorem}

\newtheorem{pro}{Proposition}
\newtheorem{notation}{Notation}
\newtheorem{defin}{Definition}

\newtheorem{exa}{Example}

\newcommand{\codefont}[1]{{\fontshape{n}\texttt{#1}}}

\newcommand{\backtick}{$\hspace{4pt}\grave{ }\hspace{2pt}$}
\newcommand{\norm}[1]{\lVert#1\rVert}

\DeclareMathOperator{\Arg}{Arg}
\DeclareMathOperator{\fix}{fix}
\DeclareMathOperator{\Fix}{Fix}
\DeclareMathOperator{\aut}{aut}
\DeclareMathOperator{\Id}{Id}

\DeclareCaptionFormat{algor}{%
  \hrulefill\par\offinterlineskip\vskip1pt%
    \textbf{#1#2}#3\offinterlineskip\hrulefill}
\DeclareCaptionStyle{algori}{singlelinecheck=off,format=algor,labelsep=space}
\newenvironment{snippet}{\vspace{10px}\fontfamily{ppl}\selectfont}{\vspace{10px}}

\let\oldsqrt\sqrt
\def\sqrt{\mathpalette\DHLhksqrt}
\def\DHLhksqrt#1#2{%
\setbox0=\hbox{$#1\oldsqrt{#2\,}$}\dimen0=\ht0
\advance\dimen0-0.2\ht0
\setbox2=\hbox{\vrule height\ht0 depth -\dimen0}%
{\box0\lower0.4pt\box2}}

\begin{document}

\title{Software for enumerative and analytic combinatorics}
\author{Andrew MacFie}
\date{2013}
\maketitle

\begin{abstract}
We survey some general-purpose symbolic software packages that implement
algorithms from enumerative and analytic combinatorics.
Software for the following areas is covered:
basic combinatorial objects,
symbolic combinatorics,
P\'olya theory,
combinatorial species, and
asymptotics.
We describe the capabilities that the packages offer as well as some of the
algorithms used, and provide links to original documentation.
Most of the packages are freely downloadable from the web.
\end{abstract}

Note: In this document, to refer to webpages
we place URL links in footnotes, and for all other types of referent we use standard endnote references.

\tableofcontents

\section{Introduction}
\label{sec:introduction}
In an opinion article\footnote{
 \url{http://www.math.rutgers.edu/~zeilberg/Opinion36.html}
} posted to his website in 1999, Doron Zeilberger challenged mathematicians to rethink the role of computers in mathematics.
``Everything that we can prove today will soon be provable, faster and better, by computers,''
he says, then gives the following advice:
\begin{quote}
The real work of us mathematicians, from now until, roughly, fifty years from now, when computers won't need us anymore, is to make the transition from human-centric math to machine-centric math as smooth and efficient as possible. \dots
 We could be much more useful than we are now, if, instead of proving yet another theorem, we would start teaching the computer everything we know, so that it would have a headstart. \dots
 Once you learned to PROGRAM (rather than just use) Maple (or, if you insist Mathematica, etc.), you should immediately get to the business of transcribing your math-knowledge into Maple.
\end{quote}

If futurist Ray Kurzweil's predictions for artificial intelligence progress\footnote{
\url{http://en.wikipedia.org/wiki/Predictions_made_by_Ray_Kurzweil}
}
are to be believed, Zeilberger's suggestions will turn out to be sound.
(However, utilitarians would urge us to consider the risks such technology would present.\footnote{
 \url{http://singinst.org/research/publications}
}\ \footnote{
 \url{http://www.existential-risk.org/}
})
What is certain even today is that those who use mathematics,
given the rise of computers,
can ask themselves if they are failing to capitalize on a productive division of labor between man and machine.
\hypertarget{started}{It}
 is no longer necessary to spend three years of Sundays factoring the Mersenne number \(2^{67}-1\), like F. N. Cole did in the 1900s \cite{cole},
and neither is it necessary to use error-prone pen and paper methods to perform an ever-growing set of mathematical procedures.

To illustrate this statement, this document examines symbolic computation (a.k.a.\ computer algebra) software packages for enumerative and algebraic combinatorics.
We start, in Section \ref{sec:Enumerativeandanalyticcombinatorics}, with an overview of the fields of enumerative and analytic combinatorics.
Then we go into more detail on the scope of the document in Section \ref{sec:notes}.
In Sections \ref{sec:basiccombinatorialobjects}--\ref{sec:asymptotics} we cover packages relating to
basic combinatorial objects,
symbolic combinatorics,
P\'olya theory,
combinatorial species, and
asymptotics.
Finally, we offer concluding observations and remarks in Section \ref{sec:conc}.

\section{Enumerative and analytic combinatorics}
\label{sec:Enumerativeandanalyticcombinatorics}
Welcome to the fields of enumerative and analytic combinatorics, a.k.a.\ combinatorial and asymptotic enumeration!
In order to completely cover what mathematicians think of when they think of these fields
(and to avoid saying ``enumerative'' or ``combinatorial''),
we break up our discussion into two parts which we call \emph{counting}, the more mathematical side, and \emph{enumeration}, the more algorithmic side.

\subsection{Counting}
Counting, the oldest mathematical subject \cite{companion}, is enumeration in the mathematical sense:  the study of the cardinalities of finite sets.
The most basic principle of counting is the addition rule \cite{feng}:

\begin{pro}[Addition rule]
If \( S \) is a set and \( A_1, A_2, \dots, A_n \) is a partition of \( S \), then
\[ |S| = |A_1| + |A_2| + \cdots + |A_n|. \]
\end{pro}

Other similar basic ways of counting
may be familiar from introductory probability,
and indeed, many concepts overlap between counting and discrete probability.\footnote{
\url{http://planetmath.org/encyclopedia/Combinatorics.html}
}

Elements of the body of work on counting can be roughly categorized based on three criteria:
whether they deal with exact or asymptotic results,
whether they speak in terms of generating functions or their coefficients, and
whether they make use of bijections or manipulations.


\subsubsection{Exact vs.\ asymptotic counting}
\begin{notation}
Boldface symbols refer to (possibly terminating) \(1\)-based sequences, i.e. \( \boldsymbol{a} = (a_1, a_2, \dots) \).%
\end{notation}

Generally, counting problems involve a triple \( (S, \boldsymbol{p}, N) \)
comprising
a countable set \(S\) of objects, a sequence \( \boldsymbol{p} = (p_1, p_2, \dots)  \) of functions \( p_i : S \rightarrow \mathbb{Z}_{\geq 0} \), and
a set \(N\) of sequences \( \boldsymbol{n} \) for which  \( \boldsymbol{p}^{-1}(\boldsymbol{n}) = \{ s \in S : p_1(s) = n_1, p_2(s) = n_2, \dots \} \subseteq S\)
 is finite.
The problem  is to answer the question
``For \(\boldsymbol{n} \in N\), how many objects does the set \( \boldsymbol{p}^{-1}(\boldsymbol{n}) \) contain?''
The definition of an \emph{exact answer} was given by Herb Wilf: a polynomial-time algorithm that computes the number \cite{companion}.

\begin{exa}
\label{exa:signature}
If \( \pi \in \mathcal{S}_n \) is a permutation, then let \( c_j(\pi) \) be the number of cycles of \( \pi \) with size \( j \).
The \emph{signature} of \( \pi \) is \( \boldsymbol{c}(\pi) = (c_1(\pi), c_2(\pi), \dots) \).
Let \( S \) be the set of all permutations, and for \( s \in S\),  let \( \boldsymbol{p}(s) \) be the signature of \( s \).
Let \( n \geq 1 \) and let \( \boldsymbol{n} = (\)\href{http://en.wikipedia.org/wiki/Iverson_bracket}{\([\)}\(j=n])_{j \geq 0} \).
Then \( |\boldsymbol{p}^{-1}(\boldsymbol{n})| = (n-1)! \).
The expression \( (n-1)! \) immediately suggests a polynomial-time algorithm to compute \( |\boldsymbol{p}^{-1}(\boldsymbol{n})| \).
\end{exa}

\sloppy Exact answers can be classified by \emph{ansatz}, meaning the form of the sequence \( (|\boldsymbol{p}^{-1}(\boldsymbol{n})|)_{ \boldsymbol{n} \in N} \), which is generally determined by the ``simplest'' reccurence relation it satisfies \cite{companion}.
\\

Exact answers are not the end of the story, however, partly because of applications to the analysis of algorithms.

In the single-parameter case, i.e.\ \( \boldsymbol{p} = (p) \) and \( N = ((0), (1), (2), \dots ) \), an asymptotic answer is a relation between \(f:n \mapsto | p^{-1}(n) | \) and ``simple'' functions, which holds as \( n \rightarrow \infty \).
Often the ``simple'' functions come from the logarithmico-exponential class \( \mathfrak{L} \) of Hardy \cite{concrete, hardy1},
and the relation is \( f(n) \sim g(n) \) as \( n \rightarrow \infty \), where \( g \in \mathfrak{L} \).
A more substantial relation is a full asymptotic series \cite{debruijn}:
\begin{defin}
\label{def:asymptoticseries}
Given a sequence of functions \( \boldsymbol{g} \) with \( g_{k+1}(n) = o(g_k(n)) \) as \( n \rightarrow \infty \) for all \(k\), and real numbers \( c_1, c_2, \dots \), the statement
\[ f(n) \sim \boldsymbol{c}\cdot \boldsymbol{g}(n) = c_1 g_1(n) + c_2 g_2(n) + c_3 g_3(n) + \cdots \]
is called an \emph{asymptotic series} for \( f \),
and it means
\begin{align*}
f(n) &= O(g_1(n)) \\
f(n) &= c_1 g_1(n) + O(g_2(n)) \\
f(n) &= c_1 g_1(n) + c_2 g_2(n) + O(g_3(n)) \\
f(n) &= c_1 g_1(n) + c_2 g_2(n) + c_3 g_3(n) + O(g_4(n))\\
&\vdots
\end{align*}
(as \( n \rightarrow \infty \).)
\end{defin}

An asymptotic answer in the multiple-parameter case is more complicated;
it generally involves (possibly just some moments of) a continuous approximation to a discrete probability distribution as one parameter approaches infinity,
or an asymptotic relation which holds as one or more parameters approach infinity at various rates.

\subsubsection{Generating functions vs.\ their coefficients}
Given a triple \( (S, \boldsymbol{p}, N) \), let \(f:N \rightarrow \mathbb{Z}_{\geq 0} \) be defined
\( f( \boldsymbol{n} ) = | \boldsymbol{p}^{-1} ( \boldsymbol{n} )| \).
A (type \(u\)) \emph{generating function} of \( S \) with \( \boldsymbol{z} = (z_1, z_2, \dots) \)
\emph{marking} \( \boldsymbol{p} \) is the element of the ring \( \mathbb{Q}[[\boldsymbol{z}]] \)
\[ F( \boldsymbol{z} ) = \sum_{\boldsymbol{n} \succeq \boldsymbol{0}} f(\boldsymbol{n}) u(\boldsymbol{n}) \boldsymbol{z}^{\boldsymbol{n}}, \]
where \( \boldsymbol{z}^{\boldsymbol{n}} = z_1^{n_1} z_2^{n_2} \cdots \).
We call \( F( \boldsymbol{z} ) \) an
\emph{ordinary} generating function iff
 \( u(\boldsymbol{n}) = 1 \) for all \( \boldsymbol{n} \succeq \boldsymbol{0} \),
and we call it an
\emph{exponential} generating function iff
 \( u(\boldsymbol{n}) = (n_1!)^{-1} \) for all \( \boldsymbol{n} \succeq \boldsymbol{0} \).

\begin{exa}
Let \( F( \boldsymbol{z}) \) be the ordinary generating function for words of length \( n \) on the alphabet \([1..k]\), with \( z_j \) marking the number of occurrences of \( j, 1 \leq j \leq k \).
We have \[ F(\boldsymbol{z}) = (z_1 + z_2 + \cdots + z_k)^n. \]
\end{exa}

To define convergence of sequences of generating functions, the norm on formal power series used in this document is defined for \( f( \boldsymbol{z} ) \neq 0 \) as
\[ \norm{f(\boldsymbol{z})} = 2^{-k}, \text{ where } k= \max \left\{ j \in \mathbb{Z}_{\geq 0} : \forall i \in [0..j], [z^i]f(z,z, \dots) = 0 \right\}, \]
and \( \norm{0} = 0 \).

It can be efficient to initially make statements about \( F( \boldsymbol{z} ) \) instead of working directly with \( f \), for a variety of reasons \cite{ac,generatingfunctionology}\footnote{
 \url{http://web.mit.edu/~qchu/Public/TopicsInGF.pdf}
}, and  extracting \( f(\boldsymbol{n}) \) exactly from a suffiently simple representation of \( F( \boldsymbol{z} ) \) can be done in polynomial time \cite{companion}.

In addition, there are very widely applicable theorems for obtaining the asymptotics of \(f\) from \( F(\boldsymbol{z}) \) \cite{ac}, which involve using the power series \( F(\boldsymbol{z}) \) to define a complex function analytic at the origin.
Since generating functions are heavily used in both the exact and asymptotic worlds, Wilf says in his book \cite{generatingfunctionology}, ``To omit the analytical (i.e.\ asymptotic) parts of [counting with generating functions] \dots is like listening to a stereo broadcast of, say, Beethoven's Ninth Symphony, using only the left audio channel.
The full beauty of the subject of generating functions emerges only from tuning in on both channels: the discrete and the continuous.''

\subsubsection{Bijective combinatorics vs.\ manipulatorics}
To prove two sets have equal cardinality, often one of two methods is used.
First, if one has algebraic expressions for the cardinality of each set, one may perform algebraic
manipulations upon them until they are syntactically equivalent.
Second, one may give an explicit bijection between them.
The following example demonstrates each method on the same problem:

\begin{exa}
A partition is called \emph{odd} if all its parts are odd and it is called \emph{distinct} if all its parts are distinct.  Let \(f(n)\) and \( g(n) \)
be the number of odd and distinct partitions of size \( n \) respectively, and let us define the ordinary generating functions \( F(z) = \sum_{n \geq 0} f(n) z^n \)
and \( G(z) = \sum_{n \geq0} g(n) z^n \).  Since
\begin{align*}
G(z) &= \prod_{n \geq 0 } (1 + z^n) \\
&= \prod_{n \geq 0} \frac{1-z^{2n}}{1-z^n} \\
&= \frac{\prod_{n \geq 0} (1-z^{2n}) }{ \prod_{ n \geq 0 } (1-z^{2n}) \prod_{n \geq 0} (1-z^{2n + 1})} \\
&= \prod_{n \geq 0} \frac{1}{1-z^{2n+1}} \\
&= F(z),
\end{align*}
we have \([z^n]G(z) = [z^n]F(z) \), and thus \( f(n) = g(n) \), for all \( n \geq 0 \).
There is also a bijective proof of this fact, due to Glaisher \cite{companion}, which we sketch.
The function from distinct partitions to odd partitions is defined as follows: Given a distinct partition, write each of its parts as \(2^r s\), where \( s \) is odd, and replace each by \(2^r\) copies of \(s\).
This function is invertible, with inverse computable as follows: Take an odd part \( a \) which occurs \( m \) times, and write
\( m \) in base \(2\), i.e.{} \(m = (s_k \cdots s_1)_2 \),
then replace the \( m \) copies of \( a \) by the
\( k \) parts \(2^{s_1}a, \dots, 2^{s_k}a \).
\end{exa}

Arguably, manipulations are not combinatorics, hence the name ``manipulatorics''.
Indeed,  here the borders between algebra, analysis, combinatorics, the analysis of algorithms, and other fields  are blurry.

Usually, bijections are used to give exact answers, but bijections can also be used in the context of asymptotics, as described in \cite{ab}, for example.\\

Standard textbooks for the field of counting include \cite{speciesbook, comtet, ac, stanley1, stanley2}.

\subsection{Enumeration}
\label{sec:enumeration}
Enumeration is the field of computer science dealing with
\emph{algorithms} that \emph{generate}
the elements of finite sets.
Various types of enumeration problems can be proposed for a given triple \( (S, \boldsymbol{p}, N) \);
the most common ones are ranking and unranking, random generation, exhaustive listing, and iteration,
which we define in that order below.

We call an \emph{ordering} of \( \boldsymbol{p}^{-1}(\boldsymbol{n}) \subseteq S \) a bijection between \( \boldsymbol{p}^{-1}(\boldsymbol{n}) \) and the integers \( [1..|\boldsymbol{p}^{-1}(\boldsymbol{n})| ] \).
A \emph{ranking} algorithm computes this bijection and an \emph{unranking} algorithm computes its inverse.

In \emph{random generation}, discrete distributions (always uniform distributions in this document) are specified on the sets \( \boldsymbol{p}^{-1}(\boldsymbol{n}) \),
and an algorithm is required to take as input \( \boldsymbol{n} \in N \) and return a random variate drawn from \( \boldsymbol{p}^{-1}(\boldsymbol{n}) \) according to the distribution.
Unranking algorithms can be used for random generation, since an integer can be generated at random and then unranked to give a random object.

An \emph{exhaustive listing} algorithm takes as input \( \boldsymbol{n} \in N \) and returns a list of all elements in
\( \boldsymbol{p}^{-1}(\boldsymbol{n}) \).
Generally, as with random generation, if the \(  \boldsymbol{p}^{-1}(\boldsymbol{n}) \) are well defined, a brute force algorithm is trivial, and the problem
lies in designing an efficient algorithm.

\emph{Iteration} is the problem of, given \( \boldsymbol{n} \in N\) and an object \(s \in \boldsymbol{p}^{-1}(\boldsymbol{n}) \), generating the next object \(s'\) in a certain ordering of \( \boldsymbol{p}^{-1}(\boldsymbol{n}) \).
It is related to the problem of exhaustive listing since any iteration algorithm immediately leads to an exhaustive listing algorithm, and more importantly, finding a particular ordering of \( \boldsymbol{p}^{-1}(\boldsymbol{n}) \)  often leads to the most efficient exhaustive listing algorithms.
If one employs an iteration algorithm repeatedly,
in an exhaustive listing algorithm,
one aims for an ordering that takes constant amortized time for each iteration.
A \emph{Gray code} is an ordering in which sucessive objects differ in some prespecified small way and thus is perfect for exhaustive listing through iteration.

\section{A few notes on the following packages}
\label{sec:notes}
In this document we focus on general-purpose software of wide interest to mathematicians, mathematics students, and perhaps those outside the field.
Many packages have been created for solving particular counting and enumeration problems, such as Lara Pudwell's enumeration schemes packages\footnote{
 \url{http://faculty.valpo.edu/lpudwell/maple.html}
}, Donald Knuth's OBDD\footnote{
 \url{http://www-cs-staff.stanford.edu/~knuth/programs.html}
} which enumerates perfect matchings of bipartite graphs, and many of Zeilberger's numerous Maple packages\footnote{
 \url{http://www.math.rutgers.edu/~zeilberg/programs.html}
};
such packages are outside the scope of this document.

There are some relatively general-purpose packages which did not make it into the document but should to be mentioned for completeness, though:
The Mathematica package Omega implements MacMahon's Partition Analysis\footnote{
 \url{http://www.risc.jku.at/research/combinat/software/Omega/index.php}
},
and there is a Maple analog written by Zeilberger called LinDiophantus\footnote{
 \url{http://www.math.rutgers.edu/~zeilberg/tokhniot/LinDiophantus}
};
an algorithm by Guo-Niu Han able to cover more general expressions than Omega was implemented in Maple\footnote{
 \url{http://www-irma.u-strasbg.fr/~guoniu/software/omega.html}
} and Mathematica\footnote{
 \url{http://www.risc.jku.at/research/combinat/software/GenOmega/index.php}
} packages;
the package RLangGFun translates from rational generating functions to regular languages\footnote{
 \url{http://www.risc.jku.at/research/combinat/software/RLangGFun/index.php}
},
and the regexpcount package from the INRIA Algorithms Group translates in the other direction\footnote{
 \url{http://algo.inria.fr/libraries/libraries.html\#regexpcount}
};
and Zeilberger has written a number of packages related to the umbral transfer matrix method, an infinite-matrix generalization of the transfer matrix method.\footnote{
 \url{http://www.math.rutgers.edu/~zeilberg/programs.html}
}

This document excludes some software related to the intersection of algebraic and enumerative combinatorics, and all software related to
power series summation and manipulation.

General-purpose software is currently skewed towards manipulations, away from bijections.
Indeed, one would not expect complicated bijections such as the proofs of
Theorem 1 in \cite{triangles} (balanced trees and rooted triangulations)
or
Theorem 4.34 in \cite{erratabook} (indecomposable 1342-avoiding \(n\)-permutations and \( \beta(0,1) \)-trees on \(n\) vertices)
to be obtainable by symbolic methods any time soon.
However, in this document, we include as many algorithms with a bijective flavor as possible.
(Ultimately all computations may be considered manipulations, but we refer to a qualitative difference in mathematical content.)

Finally, we note that plenty of basic algorithms from combinatorics, as well as advanced algorithms from the symbolic computation literature,
have not been implemented in a published package.
(Actually, whether or not an algorithm has been ``implemented in a published package'' has a fuzzy value.
For example, some have been implemented, published, but are now gone,
while others have been published and are available, but are written in obscure languages that are unfamiliar or difficult to obtain compilers for.
The sentence is true even in the loosest sense, however.)

\section{Basic combinatorial objects}
\label{sec:basiccombinatorialobjects}
\subsection{Mathematical background}
Algorithmically counting and enumerating the basic objects of combinatorics like graphs, trees, set partitions, integer partitions, integer compositions, subsets, and permutations is implemented in various packages and discussed in various books \cite{IVA, combinatorialalgorithms, NW78, skiena, W89} (and many papers).
As an example, \cite{SED77} describes over 30 permutation generation algorithms published as of 1977.

A complete comparison of the enumeration algorithms implemented in the packages in this document could be its own project.
For the packages mentioned in this document that enumerate basic objects, we do not give full details on the algorithms  used.
For more information, see the packages' documentation.
However, in the rest of this subsection we provide some examples of
two general concepts, first mentioned in Section \ref{sec:enumeration},
which guide the discovery of
enumeration algorithms: orderings and Gray codes.

\subsubsection{Orderings}
Many combinatorial objects can be ordered lexicographically.
\emph{Lexicographic order}, a.k.a.\ lex order, applies when the objects can be represented as words over an ordered alphabet.\footnote{
 \url{http://planetmath.org/encyclopedia/DictionaryOrder.html}
}
If \(w = w_1 w_2 \cdots w_n \) and \(v = v_1 v_2 \cdots v_n \) are words then, in lexicographic order, then \(w < v \) iff \(w_1 < v_1\) or there is some \( k \in [1..n-1] \) such that \(w_j = v_j \) for \( 1 \leq j \leq k \) and \( w_{k+1} < v_{k+1} \).
Permutations are a clear example of a case where this order applies, and
iterating through permutations in lexicographic order is an easy exercise, see \cite{dijkstra} for a solution.

\emph{Co-lexicographic order}, a.k.a. co-lex order, is related: If \(w\) and \(w'\) are words, \(w \leq w'\) in co-lexicographic order iff \(rev(w) \leq rev(w')\) in lexicographic order (where \(rev\) reverses words).

Another order, \emph{cool-lex order}, applies to binary words containing exactly \( k \) copies of \(1\), which we can think of as \(k\)-subsets \cite{coollex}.
Generating the next binary word in cool-lex order is done as follows: Find the shortest prefix ending in \(010\) or \(011\), or the entire word if no such prefix exists.
Then cyclically shift it one position to the right.
Since the shifted portion of the string consists of at most four contiguous runs of \(0\)'s and \(1\)'s, each succesive binary word can be generated by transposing only one or two pairs of bits.
Thus cool-lex order for \(k\)-subsets is a Gray code.

\subsubsection{Gray codes}
Unrestricted subsets have a Gray code that is very easy to understand, called the \emph{standard reflected Gray code},
in which, as above, we represent subsets as binary words.
Say we want to construct a Gray code \(G_n\) of subsets of a size-\(n\) set,
and suppose we already have a Gray code \(G_{n-1}\) of subsets of the last \(n-1\) elements of the set.
Concatenate \(G_{n-1}\) with a reversed copy of \(G_{n-1}\) with the first element of the set added to each subset.
Then all subsets differ by one from their neighbors, including the center, where the subsets are identical except for the first element.

\subsection{Combinat (Maple)}
Author: Maplesoft
\\
Website: \url{http://www.maplesoft.com/support/help/Maple/view.aspx?path=combinat}
\\

The combinat package, which is distributed with Maple, has routines for counting, listing, randomly generating, and ranking and unranking basic combinatorial objects such as permutations, \(k\)-subsets,
unrestricted subsets, integer partitions, set partitions, and integer compositions.\footnote{
 \url{http://www.maplesoft.com/support/help/Maple/view.aspx?path=combinat}
}
Like Combinatorica and unlike Sage and the Combinatorial Object Server, combinat does not offer a wide range of restrictions that can be placed on the objects.
As mentioned in Section \ref{sec:combstruct},
most of the functionality of the combinat package
is also covered by Combstruct.

Most types of objects can only be enumerated in a single ordering, but unrestricted subsets (in binary word form) can be listed in Gray code order with the
\codefont{graycode}
function.
\begin{exa}
We can use \codefont{graycode} to print all subsets of a size-\(3\) set:

\begin{snippet}
\codefont{> printf(cat(\backtick \ \%.3d\backtick \$8), op(map(convert, graycode(3), binary)))}

\codefont{000 001 011 010 110 111 101 100}
\end{snippet}
\end{exa}

We note that outside the combinat package, Maple includes support for random graph generation, which is comparable to, for example, Mathematica's.
For more information on Maple, see the Appendix.

\subsection{Combinatorial Object Server}

Author: Frank Ruskey
\\
Last modified: May 2011
\\
Website: \url{http://theory.cs.uvic.ca/cos.html}
\\

The Combinatorial Object Server (COS) is a website that runs on the University of Victoria's domain.
It has a web interface for easily specifying a set of basic combinatorial objects and viewing an exhaustive listing of all objects in the set (see Figure \ref{fig:cos} on page \pageref{fig:cos}).
Objects available include permutations, derangements, involutions, \(k\)-subsets, unrestricted subsets, set partitions, trees, necklaces, and unlabeled graphs.

On each type of object, there is a set of restrictions that can be placed.
Integer partitions, for example, can be restricted by largest part, and whether the parts must be
odd,
distinct, or odd and distinct.


There is also a wide variety of output formats for the objects.
Permutations, for example, can be printed in one line notation, cycle notation, permutation matrix form, standard Young tableau form and more.

The order of output can sometimes be specified, too.
Combinations, for example, can be shown in Gray code, lexicographic, co-lexicographic, cool-lex, transposition, or adjacent transposition orders.

\begin{figure}
\setlength\fboxsep{0pt}
\setlength\fboxrule{0.5pt}
\centering  \fbox{ \includegraphics{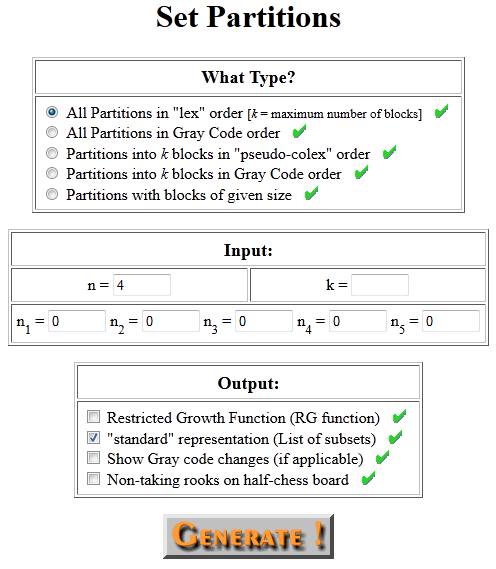} }
\caption{The Combinatorial Object Server's page for set partitions. \label{fig:cos}}
\end{figure}

\subsection{Combinatorica: basic combinatorial objects}
\label{sec:combinatoricabasicobjects}
Authors: Sriram Pemmaraju and Steven Skiena
\\
Download: \url{http://www.cs.uiowa.edu/~sriram/Combinatorica/NewCombinatorica.m}
\\
Last modified: 2006
\\
Website: \url{http://www.cs.sunysb.edu/~skiena/combinatorica/}
\\

Combinatorica is a Mathematica package for discrete mathematics.
In development since 1990, it includes over 450 functions in the areas of P\'olya theory, permutations and algebraic combinatorics, basic combinatorial objects, graph algorithms, and graph plotting.
A book was written by the package authors \cite{skiena}, which is the definitive source of information on Combinatorica.
The Combinatorica package has been included with releases of Mathematica since Mathematica version 4.2, although some of Combinatorica's functionality has recently been redone and built into the Mathematica kernel in Mathematica 8.
For information on Mathematica as a programming language, see the Appendix.\\

Combinatorica has support for counting and enumeration with permutations, \(k\)-subsets, unrestricted subsets, integer partitions, integer compositions, set partitions, Young tableaus and graphs.
For each type of object, Combinatorica generally offers rules for
counting, iteration and listing in one or two orderings,
and random generation.
Combinatorica does not provide as many ways to specify restrictions on the objects as COS or Sage.

\begin{exa}
The function \codefont{GrayCodeSubsets} exhaustively lists all subsets of a set in standard reflected Gray code order:

\begin{snippet}
\codefont{In[1]:= GrayCodeSubsets[\{1, 2, 3, 4\}]}

\codefont{Out[1]:= \{\{\},\{4\},\{3,4\},\{3\},\{2,3\},\{2,3,4\},\{2,4\},\{2\},\{1,2\},\{1,2,4\},}\\
\codefont{\{1,2,3,4\},\{1,2,3\},\{1,3\},\{1,3,4\},\{1,4\},\{1\}\}}
\end{snippet}

One may wonder if such a Gray code is unique,
and one can find this out by first noticing that Gray codes for the subsets of a size-\(n\) set are in bijection with Hamiltonian paths in the \(n\)-dimensional hypercube.
Combinatorica includes a database of common graphs, including \codefont{Hypercube[n]}, and also has the \codefont{HamiltonianCycle} rule which replaces \codefont{HamiltonianCycle[graph, All]} with a list of all Hamiltonian cycles in \codefont{graph}.
So to find out if the Gray code order above is unique, one can find the length of the list of Hamiltonian paths in the \(4\)-dimensional hypercube:\\

\begin{snippet}
\codefont{In[2]:= Length[HamiltonianCycle[Hypercube[4], All]]}

\codefont{Out[2]:= 2688}
\end{snippet}

It is definitely not!
The number of Hamiltonian cycles in an \(n\)-dimensional hypercube is not known, even asymptotically \cite{W89}.
\end{exa}

\subsection{Sage: basic combinatorial objects}

Download: \url{http://www.sagemath.org/download.html}
\\
Online access: \url{http://www.sagenb.org/}
\\
Website: \url{http://sagemath.org/doc/reference/combinat/index.html}
\\

Sage is a free, open-source computer algebra system (CAS) first released in 2005.
Sage integrates many specialized open-source symbolic and numeric packages, such as Maxima, GAP, SciPy, and NumPy, written in various languages and allows them all to be called from a unified Python interface.
In addition, it has native support for a wide and quickly expanding range of mathematical fields, including combinatorics.

This section covers Sage's capabilities for counting and enumerating basic combinatorial objects;
see
Section \ref{sec:sagespecies}
for Sage's combinatorial species capabilities.
\\

Sage uses object-oriented programming to implement a category-theoretic hierarchy of categories and objects.\footnote{
 \url{http://www.sagemath.org/doc/reference/sage/categories/category.html}
}\ \footnote{
 \url{http://www.sagemath.org/doc/reference/sage/categories/primer.html}
}
Sage's support for combinatorial objects, which is part of a migration of the MuPAD-Combinat project\footnote{
 \url{http://mupad-combinat.sourceforge.net/}
}, which has reached end-of-life, to Sage, is based on the category called \codefont{EnumeratedSets}\footnote{
 \url{http://www.sagemath.org/doc/reference/sage/categories/enumerated\_sets.html}
}.
Classes of basic combinatorial structures
(such as \(k\)-subsets, unrestricted subsets, signed and unsigned integer compositions, necklaces,
integer partitions, permutations, ordered and unordered set partitions, words and subwords)
all belong to the category \codefont{EnumeratedSets} which implies that
sets of objects from those categories can be constructed which
inherit at least the following methods:
\begin{enumerate}
 \item \codefont{cardinality()} - the cardinality of the set,
 \item \codefont{list()} - a list of all elements,
 \item \codefont{unrank(n)} - the \codefont{n}th object in an ordering,
 \item \codefont{rank(e)} - the rank of the object \codefont{e},
 \item \codefont{first()} - the first object in the ordering,
 \item \codefont{next(e)} - the next object after \codefont{e} in an ordering,
 \item \codefont{random\_element()} - an object chosen at random according to the uniform distribution.
\end{enumerate}

Of course, each class of combinatorial object built in to the system may also implement many more methods.
Many classes allow restrictions to be specified, but only the default ordering is available.





\begin{exa}
The \codefont{Partitions()} static method is called to construct an object representing a set of integer partitions specified by its arguments.\footnote{
 \url{http://sagemath.org/doc/reference/sage/combinat/partition.html}
}
For example, \codefont{Partitions(4)} returns all integer partitions of \( 4 \), while \codefont{Partitions(4, max\_part=2)} returns all partitions of \(4\) with maximum part size \(2\):

\begin{snippet}
\codefont{Partitions(4, max\_part=2).cardinality()}

\codefont{3}

\codefont{sage: Partitions(4, max\_part=2).list()}

\codefont{[[2, 2], [2, 1, 1], [1, 1, 1, 1]]}

\codefont{sage: Partitions(4, max\_part=2).random\_element()}

\codefont{[2,2] }
\end{snippet}
\end{exa}

Sage also provides several implementations of counting functions, separate from the \codefont{EnumeratedSets} category.\footnote{
 \url{http://sagemath.org/doc/reference/sage/combinat/combinat.html}
}
These include the partition-theoretic counting functions for number of set partitions, and ordered and unordered integer partitions,
and the set-theoretic counting functions for number of subsets, arrangements, derangements and permutations of a multiset.

\section{Symbolic combinatorics}
\label{sec:symboliccombinatorics}
\subsection{Mathematical background}
\label{sec:symboliccombinatoricsbackground}
Let \(S\) be a set of objects, with a parameter \(\boldsymbol{p}=(p) \).
We define a new set \(S^{<2>} = S \times S\), with parameter \( \boldsymbol{p}^{<2>} = (p^{<2>}),\) where \( p^{<2>}((s_1, s_2)) = p(s_1) + p(s_2) \) for all \( s_1, s_2 \in S\).
Let \(f(n) = |p^{-1}(n)|, f^{<2>}(n) = |(p^{<2>})^{-1}(n)|\).
Then if  \(F(z)\) is the ordinary generating function
\[ F(z) = \sum_{n \geq 0} f(n) z^n, \]
and
\( F^{<2>}(z) \) is the ordinary generating function
\[ F^{<2>}(z) = \sum_{n \geq 0} f^{<2>}(n) z^n,\]
we have, simply,
\[ F^{<2>}(z) = F(z)^2. \]
It turns out that many other correspondences exist between the structure of a set of objects and its generating function.
This document includes sections for two frameworks that develop this idea:
\emph{the theory of combinatorial species} which is the focus of Section \ref{sec:species},
and
\emph{symbolic combinatorics},
which is described below.

The central concept of symbolic combinatorics\footnote{
 \url{http://en.wikipedia.org/wiki/Symbolic\_combinatorics}
} is the combinatorial class.
\begin{defin}
A \emph{combinatorial class} is a countable set on which a parameter called \emph{size} is defined, such that the number of elements of any given size \(n \geq 0 \) is finite.
\end{defin}
If \(\mathcal{A}\) is a combinatorial class, the size of an element \(\alpha \in \mathcal{A} \) is denoted \( |\alpha| \).
We denote the set of elements of size \(n\) in \(\mathcal{A}\) by \(\mathcal{A}_n\), and denote its cardinality by \( a_n = |\mathcal{A}_n| \).
\begin{defin}
The \emph{counting sequence} of a combinatorial class \(\mathcal{A}\) is the sequence \( (a_n)_{n\geq0} \).
\end{defin}

There are two types of combinatorial class, \emph{unlabeled} and \emph{labeled}.

\subsubsection{Unlabeled classes}
The word \emph{class} in this section refers to an unlabeled combinatorial class, which can be thought of as a set of objects made up of nodes without unique labels (think graphs).
This will become rigorous as we proceed.

\begin{defin}
The (ordinary) \emph{generating function} of a class \(\mathcal{A}\) with \(z\) marking size is the formal power series
\[ A(z) = \sum_{n \geq 0} a_n z^n.\]
\end{defin}
\begin{defin}
A (\(k\)-ary) \emph{combinatorial construction} \(\Phi\) is a function that maps combinatorial classes \(\mathcal{B}^{<1>}, \mathcal{B}^{<2>}, \dots, \mathcal{B}^{<k>}\) to a new class \(\mathcal{A} = \Phi(\mathcal{B}^{<1>}, \mathcal{B}^{<2>}, \dots, \mathcal{B}^{<k>}) \).

The combinatorial construction \( \Phi \) is \emph{admissible} iff the counting sequence of \( \mathcal{A} \) only depends on the counting sequences of the arguments \( \mathcal{B}^{<1>}, \mathcal{B}^{<2>}, \dots, \mathcal{B}^{<k>} \).
\end{defin}

If a construction \( \Phi \) is admissible, there exists a corresponding operator \(\Psi\) on generating functions such that if
\(\mathcal{A} = \Phi(\mathcal{B}^{<1>}, \mathcal{B}^{<2>}, \dots, \mathcal{B}^{<k>}), \)
 then
\[A(z) = \Psi(B^{<1>}(z), B^{<2>}(z), \dots, B^{<k>}(z) ). \]

The basic admissible constructions for unlabeled classes are called sum, product, sequence, powerset, multiset, and cycle.
The definitions and corresponding generating function operators for all of these can be found in \cite{ac}; here we only describe the first three.

\begin{enumerate}
\item The \emph{sum} of two classes \(\mathcal{A}\) and \(\mathcal{B}\) is written \( \mathcal{A} + \mathcal{B} \) and is formed by the discriminated union\footnote{
 \url{http://en.wikipedia.org/wiki/Disjoint\_union}
}
of \( \mathcal{A} \) and \( \mathcal{B} \), with size inherited from the summands.
The generating function of \( \mathcal{A} + \mathcal{B} \) is \(A(z) + B(z) \).

\item The \emph{product} of two classes is written \( \mathcal{A} \times \mathcal{B} \) and is formed by the cartesian product of \( \mathcal{A} \) and \( \mathcal{B} \), with size defined additively.  This is the construction used \hyperref[sec:symboliccombinatoricsbackground]{above}, where we saw that the generating function for \( \mathcal{A} \times \mathcal{B} \) is \(A(z)B(z)\).

\item Finally, the \emph{sequence} construction \( \textsc{Seq} \) is defined on classes with no elements of size \(0\).
For a class \( \mathcal{A} \), the value \( \textsc{Seq}( \mathcal{A}) \) is the set of all finite sequences of elements in \( \mathcal{A} \), with size defined additively.  The generating function for  \( \textsc{Seq}( \mathcal{A}) \)  is
\[ 1 + A(z) + A(z)^2 + \cdots = \frac{1}{1-A(z)}. \]
\end{enumerate}

The basic combinatorial constructions can be modified (\emph{restricted}) in a number of ways.
For example, we can fix \(k\) and define a construction \( \textsc{Seq}_{ \geq k} \) that constructs sequences of length at least \(k\),
or for another example we could construct products containing an element of even size from one class and an element of odd size from another, etc.

Let \( \mathcal{E} \) be the class with a single element \(\epsilon\) of size \(0 \),
called the \emph{neutral object},
and let \(\mathcal{Z}\) be the class with a single element \( \zeta \) of size \(1\),
called an \emph{atom}.

\begin{defin}
A \emph{specification} for an \(r\)-tuple of classes is a collection of \(r\) equations
\begin{align*}
\mathcal{A}^{<1>} &= \Phi_1(\mathcal{A}^{<1>}, \dots, \mathcal{A}^{<r>}) \\
\mathcal{A}^{<2>} &= \Phi_2(\mathcal{A}^{<1>}, \dots, \mathcal{A}^{<r>}) \\
&\cdots \\
\mathcal{A}^{<r>} &= \Phi_r(\mathcal{A}^{<1>}, \dots, \mathcal{A}^{<r>}),
\end{align*}
where each \(\Phi_i(\cdots)\) represents an expression built from the \(\mathcal{A}\)'s using the (possibly restricted) basic admissible constructions, as well as the classes \(\mathcal{E}\) and \(\mathcal{Z}\).
\end{defin}

\begin{exa}
Let \( \mathcal{T} \) be the class of nonempty unlabeled plane trees, with the size of a tree being the number of nodes.
Then \( \mathcal{T} \) satisfies the specification
\[ \mathcal{T} =  \mathcal{Z} \times \textsc{Seq}(\mathcal{T}), \]
since an object in \( \mathcal{T} \) is a single root with a sequence of subtrees.
This specification implies \(T(z) = z (1-T(z))^{-1} \), and thus \( T(z) = \frac{1}{2} \left(1-\sqrt{1-4 z}\right) \).
\end{exa}

\subsubsection{Labeled classes}
An object in a labeled class is \emph{labeled}, meaning, if it has size \( n \),
each of its \( n \) indivisible components is labeled with a unique integer from the set \( [1..n] \).
A rigorous way to define such classes begins with a different definition for the elementary classes:
Again, let \( \mathcal{E} \) be the class with one neutral object \(\epsilon\) of size \(0\),
but now let \( \mathcal{Z} \) be the class containing one labeled element \( \zeta \) of size \(1\), a \emph{labeled atom}.
Then,
labeled classes can be defined by specifications as above if we define some useful constructions on labeled classes.

There is indeed a set of admissible constructions for labeled classes that is analogous to that for the unlabeled case.
The sum of two classes is defined the same as for unlabeled classes, but for the others, we first need to define the product of two labeled objects:
\begin{defin}
Given two labeled objects, \( \alpha \) and \( \beta \), the \emph{labeled product} \( \alpha \star \beta \) is the set of all pairs \( ( \alpha', \beta' ) \) where \(\alpha' \) and \( \beta' \) are relabeled versions of \(\alpha\) and \( \beta \) such that order is perserved in the relabelings and each number in \([1..|\alpha| + |\beta|]\) appears as a label in either \(\alpha'\) or \( \beta' \).
\end{defin}
This concept leads to the definition of the product, sequence, set and cycle constructions, the first of which we define here; the rest can be found in \cite{ac}.
The product of labeled classes \(\mathcal{A} \) and \( \mathcal{B} \) is the set
\[ \mathcal{A} \star \mathcal{B} = \bigcup_{\alpha \in A, \beta \in B} (\alpha \star \beta), \]
with size defined additively.

As with the unlabeled case, the usefulness of defining a labeled class in terms of a specification comes from the fact that labeled constructions correspond to relatively ``simple'' operators on generating functions --- exponenential generating functions in the labeled case.
\begin{defin}
The (exponential) \emph{generating fuction} of a labeled class \(\mathcal{A}\) with \(z\) marking size is the formal power series
\[ A(z) = \sum_{n \geq 0} a_n \frac{z^n}{n!}. \]
\end{defin}
\noindent For example, the generating function for \(\mathcal{A} \star \mathcal{B} \) is \( A(z)B(z) \).

\begin{exa}
We define a labeled binary tree as a labeled tree in which every internal node has two children.
Let \( \mathcal{B} \) be the labeled class representing such trees.
Since an object in \( \mathcal{B} \) is either a node with no children or a node with two children, we have
\[ \mathcal{B} = \mathcal{Z} + \mathcal{Z} \star \mathcal{B} \star \mathcal{B}, \]
which implies that \( B(z) = z + z B(z)^2 \), and thus \(B(z) = (1-\sqrt{1-4 z^2})/2 z \).
\end{exa}

\subsubsection{Multiple parameters}
Unlabeled and labeled classes can be augmented with parameters other than size.
For example, if there is one more parameter \( \varphi \), we can redefine the original elementary classes so that \( \mathcal{E} = \{ \epsilon \} \), where \( |\epsilon| = \varphi(\epsilon) = 0 \), and \( \mathcal{Z} = \{ \zeta \} \), where \( |\zeta| = 1, \varphi(\zeta)=0 \), and define a new \( \varphi \)-\hypertarget{atomic}{atomic}
 class \( \mathcal{P} = \{ \pi \} \), where \( |\pi| = 0, \varphi(\pi) = 1 \).
Size is the only parameter with respect to which a structure is ``labeled'' or ``unlabeled'' (and the only parameter which may be marked by a variable with a factorial below it in the generating function for the class), so those words can still be used unambiguously to refer to a class with more than one parameter.

All constructions for labeled and unlabeled classes discussed so far have defined size additively, e.g. the size of an object is the sum of the sizes of its components.
All of these constructions can be defined on multi-parameter classes with non-size parameters defined additively, just like size.
The generating function equations they correspond to are the same as the single-parameter ones, except the generating functions may be multivariate (with a different variable marking each parameter).

\begin{exa}
	Let \( \mathcal{A} \) be a single-parameter unlabeled combinatorial class.
	One can define a new class, \( \mathcal{B} \), consisting of sequences of elements of \( \mathcal{A} \), with size and an additional parameter \( \varphi \) such that \( \varphi( \beta ) \) is the number of elements of \( \mathcal{A} \) in \( \beta \), for all \( \beta \in \mathcal{B} \).
	Then \( \mathcal{B} = \textsc{Seq}(\mathcal{P} \times \mathcal{A} ) \), and, if \(u\) marks \( \varphi \), \( A(z,u) = ( 1 - u B(z) )^{-1} \).
\end{exa}

For more details, and for information on constructions where additional parameters are not defined additively, see \cite{ac}.
\\

Given a labeled or unlabeled class \( \mathcal{A} \) with additional parameter \( \varphi \), one can define a random variable \( X_n \) to be an object chosen at random from \( \mathcal{A}_n \) according to the uniform distribution.
If \( u \) marks \( \varphi \) in the bivariate generating function \(A(z, u) \), then we have the syntactically simple relation
\[
	\mathbf{E}[\varphi(X_n)] = \frac{[z^n] A^{(0,1)}(z,u)|_{u=1}}{[z^n] A(z,1)}.
\]

Higher factorial moments are obtained similarly.
Techniques for obtaining limiting distributions from multivariate generating functions also exist: see \cite{ac}.
However, note that these methods for obtaining probabilistic facts apply for any multivariate generating function, whether or not it was obtained with symbolic combinatorics.

\subsection{Combstruct}
\label{sec:combstruct}

Authors: INRIA Algorithms Group\footnote{
 \url{http://algo.inria.fr/index.html}
} and Maplesoft
\\
Website:\\
\indent \url{http://www.maplesoft.com/support/help/Maple/view.aspx?path=combstruct}
\\

Combstruct is a Maple package
originally developed by the INRIA Algorithms Group
which is now distributed with the most recently released version of Maple, Maple 16.
Its functionality has changed over time, but
today, it includes the capabilities of ALAS from LUO (Section \ref{sec:alas}) along with the ability to enumerate, both randomly and exhaustively, the objects of a given size from a specified combinatorial class.
Combstruct also extends ALAS by supporting translation from multiple-parameter specifications to mulivariate generating functions.
Combstruct can also count, randomly and exhaustively enumerate, and iterate through
a small set of built-in, predefined structures.
In fact, Combstruct provides most of the functionality of Maple's combinat\footnote{
 \url{http://www.maplesoft.com/support/help/Maple/view.aspx?path=combinat}
} package for working with basic combinatorial structures, but with syntax that is unified with that for working with classes of objects defined by the user by combinatorial specifications.
We elaborate on these areas of functionality in the rest of this section.

Combstruct allows the user to create labeled or unlabeled single-parameter combinatorial specifications, and augment the specifications with additional parameters separately.
To specify a class, the elementary classes \( \mathcal{E} \) and \( \mathcal{Z} \) can be used, as well as the constructions
sum, product, set, powerset, sequence, cycle and substitution.
The constructions set, powerset, sequence and cycle can be restricted by an inequality or equality relation on the number of components allowed in the objects.
It is possible, of course, to define a great many types of basic combinatorial object, as well as more complicated objects, with such specifications.

The \codefont{gfeqns} command returns the system of equations over generating functions corresponding to a well-defined single-parameter specification (see Definition \ref{defin:welldefined}), and \codefont{gfsolve} attempts to return explicit expressions for the generating functions.
The \codefont{gfseries} command returns the initial values of the counting sequences of the classes.

\begin{exa}

We create a specification \codefont{bintreespec} for labeled binary trees, where \codefont{B} is the class of labeled binary trees and \codefont{Z} is the atomic class:

\begin{snippet}
\codefont{>bintreespec := \{B=Union(Z, Prod(Z, B, B)), Z=Atom \}:}
\end{snippet}

(The \codefont{Union} construction is the same as sum.)
Then we use \codefont{gfsolve} to get the generation functions:

\begin{snippet}
\codefont{>gfsolve(bintreespec, labeled, z)}

\codefont{{B(z) = -(1/2)*(-1+sqrt(1-4*z\^{}2))/z, Z(z) = z}}
\end{snippet}

\end{exa}

Combstruct's \codefont{draw} command takes a single-parameter specification, the name of a class \( \mathcal{A} \), and an integer \( n \geq 0 \) and returns a object chosen at random from the set \( \mathcal{A}_n \) according to the uniform distribution.
For information on the algorithms used, see the documentation.

To add additional parameters to a single-parameter specification, Combstruct allows the user to use an attribute grammar (see the documentation and \cite{attributegrammars} for more information on attribute grammars; also, note that it is also possible to augment single-parameter specifications in Combstruct using a more limited method
based on defining and using new atomic classes as described in the \hyperlink{atomic}{Mathematical background}).
The \codefont{agfeqns} command returns the system of equations over multivariate generating functions corresponding to a  given specification and attribute grammar.
The \codefont{agfseries} command returns the initial values of the multidimensional sequences.
The \codefont{agfmomentsolve} command takes an integer \(k \geq 0 \) and a set of equations over \( n \) multivariate generating functions \(A^{<i>}(z, \boldsymbol{u}), 1 \leq i \leq n \) and attempts to return explicit expressions for \( D_{\boldsymbol{u}}^k A^{<i>}(z, \boldsymbol{u})|_{\boldsymbol{u} = \boldsymbol{1}}, 1 \leq i \leq n\).

A number of in-depth examples of Combstruct in action are available at the INRIA Algorithms Group's website\footnote{
 \url{http://algo.inria.fr/libraries/autocomb/}
}.

\subsection{Encyclopedia of Combinatorial Structures}
\label{sec:encyclopedia}

Authors: Fr\'ed\'eric Chyzak, Alexis Darrasse, and St\'ephanie Petit
\\
Download: \url{http://algo.inria.fr/libraries/\#down}
\\
Last modified: July 2000 --- original version\footnote{
 \url{http://algo.inria.fr/salvy/index.html}
},
2011 --- online version\footnote{
 \url{http://algo.inria.fr/encyclopedia/intro.html}
}
\\

The Encyclopedia of Combinatorial Structures started out as a Maple package written by St\'ephanie Petit as part of the INRIA Algorithm Group's algolib and in 2009 a web interface for it was created by Alexis Darrasse and Fr\'ed\'eric Chyzak at \url{http://algo.inria.fr/encyclopedia/intro.html}.

The Encyclopedia is a database of counting sequences of specifiable combinatorial structures.
For each sequence in the database, the following fields, if available, are either computed or stored:
\begin{enumerate}
\item Name
\item Combinatorial specification, in \hyperref[sec:combstruct]{Combstruct} syntax
\item Initial values, obtained with Combstruct's \codefont{count}
\item Generating function, obtained with Combstruct's gfsolve
\item A linear reccurrence relation, if applicable, obtained with gfun's \codefont{holexprtodiffeq} and \codefont{diffeqtorec}
\item \sloppy A closed-form expression, obtained with either Maple's \codefont{rsolve} or gfun's \codefont{ratpolytocoeff}
\item Dominant asymptotic term as \( n \rightarrow \infty \) computed by \hyperref[sec:gdev]{gdev}'s \codefont{equivalent}
\item Description of combinatorial structure
\item References, such as entry in Sloane's Encyclopedia of Integer Sequences\footnote{
 \url{http://www.research.att.com/~njas/sequences/}
}
\end{enumerate}
(Note that gfun is the name of a Maple package developed by the INRIA Algorithms Group, for more information, see \footnote{
 \url{http://algo.inria.fr/libraries/\#gfun}
}.)
It is possible to search the database by initial values of the sequence, keywords, generating function, or closed form of the sequence.

\subsection{Lambda-Upsilon-Omega (LUO): symbolic combinatorics}
\label{sec:alas}

Authors: Bruno Salvy and Paul Zimmermann
\\
Download: \url{http://www.loria.fr/~zimmerma/software/luoV2.1.tar.gz}
\\
Last modified: May 1995\footnote{
 \url{http://algo.inria.fr/salvy/index.html}
}
\\
Website: \url{http://algo.inria.fr/libraries/libraries.html\#luo}
\\

LUO is a software project started in the late 1980s designed to automatically analyze algorithms.
It is no longer heavily used;
most of its functionality is available in \hyperref[sec:combstruct]{Combstruct} and the \codefont{equivalent} command in \hyperref[sec:gdev]{gdev}.\footnote{
 \url{http://algo.inria.fr/salvy/index.html}
}
In this document, we focus on the subset of its capabilities related to enumerative and analytic combinatorics,
omitting a discussion of its capabilities for algorithm analyis.

\hypertarget{alas}{LUO}
is made up of two modules: the algebraic analyzer (ALAS) and the analytic analyzer (ANANAS).
ALAS takes as input a combinatorial specification, either labeled or unlabeled, and outputs a system of equations over generating functions for the classes in the specification.
Then an intermediate process attempts to solve the equations explicitly; if successful, it passes the solutions to ANANAS.
ANANAS then employs a routine (which became \codefont{equivalent} in gdev) on the expressions which returns asymptotic expressions for the coefficients.
ALAS is described in further detail below, and ANANAS is described in Section \ref{sec:ananas}

Official documentation for LUO comes in the form of a main article discussing the algorithms used, but not the code \cite{assistant}, a cookbook containing a summary of \cite{assistant} and a selection of examples of algorithms being analyzed \cite{cookbook}, and some notes on the code and usage of ALAS, including a reference for writing programs in the syntax that the system can analyze, and examples of ALAS in action \cite{alas}.\\

As mentioned
\hyperlink{alas}{above},
ALAS, written by Paul Zimmermann, is a tool for translating specifications to systems of equations over generating functions.
That is, the user supplies a list of \(r\) equations
\begin{align*}
\mathcal{A}^{<1>} &= \Phi_1(\mathcal{A}^{<1>}, \dots, \mathcal{A}^{<r>}) \\
\mathcal{A}^{<2>} &= \Phi_2(\mathcal{A}^{<1>}, \dots, \mathcal{A}^{<r>}) \\
&\cdots \\
\mathcal{A}^{<r>} &= \Phi_r(\mathcal{A}^{<1>}, \dots, \mathcal{A}^{<r>}),
\end{align*}
where each \(\Phi_i(\cdots)\) represents an expression built from the \(\mathcal{A}\)'s using the basic admissible constructions, as well as the classes \(\mathcal{E}\) and \(\mathcal{Z}\).
Some restricted constructions are also allowed.

\begin{defin}
The \emph{valuation} of a class \( \mathcal{A}^{<i>} \) is the minimum size of an object in \( \mathcal{A}^{<i>} \) built according to the specification.
\end{defin}

\begin{defin}
\label{defin:welldefined}
A combinatorial specification is \emph{well defined} iff it satisfies the two properties
\begin{enumerate}
\item each class has finite valuation, and
\item for each class \( \mathcal{A}^{<i>} \) and \( n \geq 0 \), the number of objects of size \(n\) in \( \mathcal{A}^{<i>} \) built according to the specification is finite.
\end{enumerate}
\end{defin}
Since there is nothing preventing the user from supplying a non-well defined specification,
it would be nice if the well-definedness of a specification could be checked programmaticly and indeed it can and ALAS does this. See \cite{assistant} for details.

Once ALAS verifies that the specification is well-defined, it proceeds to generate the corresponding equations over generating functions using simple replacement rules.
These generating function equations can then be used to (among other things) compute initial values of the counting sequences.

The following theorem is proved in \cite{assistant}:
\begin{thm}
The number of arithmetic operations necessary for computing all the counting sequences associated with all classes \( \mathcal{A}^{<i>} \) up to size \(n\) is \(O ( rn^2 ) \).
\end{thm}

\section{P\'olya theory}

\subsection{Mathematical background}

P\'olya theory is the study of counting symmetric objects, which makes use of group theory and generating functions.
It was originally developed by John Redfield, then refounded by George P\'olya, who used it to count chemical compounds, among other objects, and after whom the central theorem (Theorem \ref{thm:polya} on page \pageref{page:polyatheorem}) is named.

\begin{defin}
An \emph{action} of a group \( H \) on a set \( X \) is a homomorphism \( \theta: H \rightarrow \mathcal{S}(X) \) from \( H \) to the symmetric group on X.
\end{defin}

In this section, the groups we work with are assumed to be the images of group actions on finite sets, i.e.\ permutation groups.
Let \( G \) be such a group on a set \( X \).
Define an equivalence relation \( \equiv \) on \( X \) by
\[ x \equiv y \text{ iff } \exists g \in G : g(x) = y. \]
The set of \hypertarget{quotient}{equivalence classes} of \( X \) under \( \equiv \) is written \( X / G \).

Elements of \( X / G \) are called the \emph{orbits} of \( X \) under \( G \),
and the \emph{orbit} of an element \( x \in X \) is the equivalence class represented by \(x\) and is denoted by \( G(x) \).
The \emph{stabilizer} \( G_x \) of an element \( x \in X \) is
 \[ G_x = \{ g \in G : g(x) = x \} . \]

\vspace{10px}

Our first lemmas bring counting into the picture:
\begin{lem}[Orbit-stabilizer]
\label{lem:orbit}
If \( G \) is a permutation group on \( X \) then for all \(x \in X\),
\[ |G(x)| = |G|/|G_x|. \]
\end{lem}
\begin{proof}
 For any \( x \in X \), consider the mapping \( f: G \rightarrow X \) defined \( g \mapsto g(x) \).
Then there is a bijection between the image of \( f \), which is  \(G(x) \),
and the set of left cosets of \( G_x \), which has cardinality \( |G| / |G_x| \),
given by \( h(x) \mapsto hG_x \) for all \( h(x) \in G(x) \).
\end{proof}

\begin{lem}[Burnside]
\label{lem:burnside}
The number of orbits of a permutation group \( G \) on a set \( X \) is
\[ | X / G | = \frac{1}{|G|} \sum_{g \in G} \fix(g), \]
where \( \fix(g) \) is the number of fixed points of \( g \).
\end{lem}
\begin{proof}
\begin{align*}
\sum_{g \in G} \fix(g) &= \sum_{g \in G} \sum_{x \in X} [g(x) = x] \\
&= \sum_{x \in X} \sum _{g \in G} [g(x) = x] \\
&= \sum_{x \in X} | G_x |\\
&= \sum_{x \in X} \frac{|G|}{|G(x)|} \\
&= |G| \sum_{x \in X} \frac{1}{|G(x)|} \\
&= |G| \sum_{A \in X/G} \sum_{x \in A} \frac{1}{|A|} \\
&= |G| \sum_{A \in X/G} |A| \frac{1}{|A|} \\
&= |G| |X / G| \qedhere
\end{align*}
\end{proof}

\vspace{10px}

\begin{defin}
A \emph{graph automorphism} of a graph \( \mathfrak{G} = (V,E) \) is a bijection \( \pi : V \rightarrow V \) such that \( \{ (\pi(v_1), \pi(v_2)) | (v_1, v_2) \in E \} = E \).
The set of all automorphisms of \( \mathfrak{G} \) forms a group, \( \aut(\mathfrak{G}) \).
\end{defin}

\begin{exa}
Say we are given a labeled graph \( \mathfrak{G} \), and we would like to find the number of relabelings (bijections on the vertex set) that yield distinct graphs.
Let \( X \) be set of all relabelings, and let \( \theta \) be the induced action of \( \aut(\mathfrak{G}) \) on \( X \), i.e. if \( \pi \in X \) and \( g \in \aut(\mathfrak{G}) \), then
\( \theta(g)(\pi)  = ( g(\pi_1), g(\pi_2), \dots) \).
Then the number of such relabelings, which is the number of isomorphic graphs on the same vertex set, is \( \left|X / \theta ( \aut(\mathfrak{G}) ) \right| \).

(The packages nauty\footnote{
 \url{http://cs.anu.edu.au/~bdm/nauty/}
},
saucy\footnote{
 \url{http://vlsicad.eecs.umich.edu/BK/SAUCY/}
}
and bliss\footnote{
 \url{http://www.tcs.hut.fi/Software/bliss/}
}, as well as some of the packages in this section,
can compute the automorphism group of a graph.
The code that nauty uses is also included in GRAPE.\footnote{
 \url{http://www.gap-system.org/Packages/grape.html}
})
\end{exa}


We now begin working up to the central theorem of P\'olya theory: a generalization of Lemma \ref{lem:burnside} for counting weighted colorings of an object.

\begin{defin}
 Let \(g \in G\) be a permutation, and let \( \boldsymbol{c}(g) \) be the \hyperref[exa:signature]{signature} of \( g \).
 Then the \emph{cycle index} of \(g\) is the formal monomial
 \[ z(g; \boldsymbol{s} ) = \boldsymbol{s}^{\boldsymbol{c}(g) }. \]
The cycle index of
the whole permutation group
\( G \) is the terminating formal power series
\[ Z(G; \boldsymbol{s}) = \frac{1}{|G|} \sum_{g \in G} z(g; \boldsymbol{s}). \]
\end{defin}

Let \( \Phi = \{ \phi_1, \phi_2, \dots \} \) be a countable set of ``colors'', each of which has a non-negative weight \(w(\phi_i) \),
such that \( w^{-1}(n) \) is finite for all \( n \geq 0 \).
We define the color-counting generating function
\[ a(t) = \sum_{ n \geq 0 } a_n t^n, \]
where \( a_n = | w^{-1}(n) | \).

Let \( \mathcal{F} = \Phi^X \) be the set of functions from \( X \rightarrow \Phi \);
we call these functions \emph{colorings}.
The \emph{total weight} of a coloring \( f \in \mathcal{F} \) is \( \sum_{x \in X} w(f(x)) \).
Define an action \( \theta \) of \( G \) on \( \mathcal{F} \) by
\[ g \mapsto ( f \mapsto (f \circ g^{-1})), \]
and let the coloring-counting generating function be
\[ b(t) = \sum_{n \geq 0} b_n t^n, \]
where \( b_n \) is the number of orbits of \( \theta(G) \) on \( \mathcal{F} \) with total weight \( n \).

\begin{lem}
\label{lem:polya}
Say \(|X| = m \).
Then the generating function for functions from \(X\) to \( \Phi \) fixed by a permutation \( \theta(g) \) with total weight marked by \( t \) is
\[ z(g; a(t), a(t^2), \dots, a(t^m) ) . \]
\end{lem}
\begin{proof}
A function is fixed by \( \theta(g) \) iff it is constant on the cycles of \( g \).
So a function fixed by \( \theta(g) \) is specified by, for each \( i \), a mapping between each cycle of size \( i \) and \( \Phi \).
For each cycle size \( i \), these mappings have generating function \( a(t^i)^{c_i(g)} \) since there are \(c_i(g) \) cycles of size \( i \) and
each color of weight \( k \) contributes a function of weight \( ki \) that is constant on the \(i\)-cycle.
The overall generating function is thus
\[ a(t)^{c_1(g)} a(t^2)^{c_2(g)} \cdots a(t^m) = z(g; a(t), a(t^2), \dots, a(t^m) ). \qedhere \]
\end{proof}

\begin{thm}[P\'olya enumeration theorem, single variable version]
\label{page:polyatheorem}
\label{thm:polya}
Say \(|X| = m \), then
\[ b(t) = Z(G; a(t), a(t^2), \dots, a(t^m)). \]
\end{thm}
\begin{proof}
Immediate from Lemmas \ref{lem:burnside} and \ref{lem:polya}.
\end{proof}

\begin{exa}
A \emph{necklace} of size \(n \) is a cycle of colored beads that can be flipped over or rotated and still be considered the same object.
In this example, we count black and white necklaces of size \( 4 \) with a given number of black beads.
We can use as our set \( X = \{1,2,3,4\} \), and use as our permutation group on \( X \) the dihedral group \( \mathcal{D}_4 \) with \( 8 \) elements.
We can let weight be the number of black beads, so that \( a(t) = 1 + t \).
Since \( \mathcal{D}_4 \) has cycle index
\[ Z(\mathcal{D}_4; s_1, s_2, s_3, s_4) = \frac{s_1^4+2 s_1^2 s_2+3 s_2^2+2 s_4}{8}, \]
we have
\[ b(t) = Z(\mathcal{D}_4; 1+t, 1+t^2, 1+t^3, 1+t^4) = 1+t+2 t^2+t^3+t^4. \]
From this generating function, we can also read off that there are \(6\) necklaces of size \( 4 \) with (at most) two colors of bead, as shown in Figure \ref{fig:necklaces}, taken from MathWorld\footnote{
 \url{http://mathworld.wolfram.com/Necklace.html}
}.

\begin{figure}
\centering  \includegraphics[width=400px]{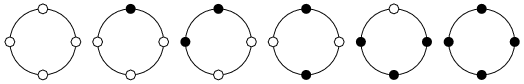}
\caption{The six necklaces with four beads and two colors. \label{fig:necklaces}}
\end{figure}
\end{exa}

\begin{exa}
\label{exa:graphs}
 Suppose we would like to know the number of unlabeled graphs with \( 3 \) vertices and \( k \) edges, \( 0 \leq k \leq 3 \), up to isomorphism.
 P\'olya's enumeration theorem can be used if we represent unlabeled graphs with \( 3 \) vertices as \(2\)-colorings of the set \( E \) of all \( \binom{3}{2} \) possible edges, where the color black represents an edge and the color white represents no edge.
 Total weight of a coloring equals the number of black edges, so our color-counting generating function is \(a(t) = 1+t \).
 The group of permutations on \( E \) that we need to quotient out is isomorphic to \(\mathcal{S}_3\), the symmetric group on the three vertices.
 The group \( \mathcal{S}_3 \) has cycle index
\[ Z( \mathcal{S}_3; s_1, s_2, s_3) = \frac{s_1^3 + 3 s_1 s_2 + 2 s_3}{3!}, \]
thus
\[ b(t) = Z( \mathcal{S}_3; 1+t, 1+t^2, 1+t^3)  = 1 + t + t^2 + t^3, \]
and we see that there is exactly one graph, up to isomorphism, with \( 3 \) vertices and \( k \) edges, \( 0 \leq k \leq 3 \).
Each of these is shown in Figure \ref{fig:graphs}.

\begin{figure}
\centering  \includegraphics[width=272px]{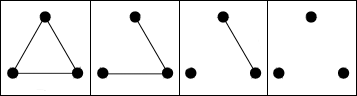}
\caption{The four unlabeled graphs with \( 3 \) vertices. \label{fig:graphs}}
\end{figure}

\end{exa}

For more information on P\'olya theory, a more general version of Theorem \ref{thm:polya}, and many more examples, see \cite{speciesbook,polya}\footnote{
 \url{http://en.wikipedia.org/wiki/P\%C3\%B3lya_enumeration_theorem}
}.

\subsection{A note on P\'olya theory packages}
Our goal in this section is to describe packages designed for P\'olya theory-type counting, rather than general computational group theory.
Capable software for the latter includes GAP, Magma, Maple, Mathematica, and Sage.
In the words of the authors of Combinatorica, ``Our aim in introducting permutation groups into Combinatorica is primarily for solving combinatorial enumeration problems.
We make no attempt to efficiently represent permutation groups or to solve many of the standard computational problems in group theory.'' \cite{skiena}

\subsection{COCO}
Download: \url{http://www.win.tue.nl/~aeb/ftpdocs/math/coco/coco-1.2a.tar.gz}
\\

COCO is a package for doing computations with permutation groups which was designed to investigate \textbf{co}herent \textbf{co}nfigurations, which are a certain type of edge--colored complete graphs \cite{cocopaper}.
COCO includes routines for, among other things,
finding the automorphism group of an edge--colored complete graph and,
given a base set and a permutation group on that set,
computing the induced permutation group on a set of combinatorial structures over the base set.
For more information, see its documentation and \cite{cocopaper}.

\subsection{Combinatorica: P\'olya theory}

Authors: Sriram  Pemmaraju and Steven Skiena
\\
Download: \url{http://www.cs.uiowa.edu/~sriram/Combinatorica/NewCombinatorica.m}
\\
Last modified: 2006
\\
Website: \url{http://www.cs.sunysb.edu/~skiena/combinatorica/}
\\

Combinatorica has P\'olya-theoretic functionality that integrates with the rest of its capabilities, such as basic combinatorial objects, a description of which, along with an overview of the Combinatorica package can be found in Section \ref{sec:combinatoricabasicobjects}.\\

Combinatorica has built-in  rules representing the symmetric, cyclic, dihedral and alternating groups, as well as the ability to create groups from other groups or simply from a set of permutations.
Combinatorica can compute the automorphism group of a graph.

The rules \codefont{Orbits} and \codefont{OrbitRepresentatives} take a set and a permutation group, and optionally how the group acts on the set (the default action being the identity map), and return the set of orbits and representatives of those orbits, respectively.

\begin{exa}
The \codefont{OrbitRepresentatives} rule can be used to
list all distinct necklaces of size \( 4 \) with \( 2 \) colors.
To do so, we evaluate the rule with the group \codefont{DihedralGroup[4]}, which is \(D_4\), and the set of words of length \( 4 \) over the letters \(R\) and \(B\):

\begin{snippet}
\codefont{In[1]:= OrbitRepresentatives[DihedralGroup[4], Strings[\{R, B\}, 4]] }

\codefont{Out[1]:= \{\{R,R,R,R\},\{B,B,B,B\},\{B,B,B,R\},\{B,B,R,R\},\{B,R,B,R\},\{B,R,R,R\}\} }
\end{snippet}
\end{exa}

The rule \codefont{CycleStructure} gives the cycle index of a single permutation, and \codefont{CycleIndex} gives the cycle index of a permutation group.

Combinatorica comes with special rules for cycle index of symmetric, alternating, cyclic and dihedral groups which work much faster than \codefont{CycleIndex}.

For P\'olya's enumeration theorem, we have the rule \codefont{OrbitInventory}, which takes the cycle index \( G( Z; \boldsymbol{s} ) \) of a group \(G\) and a list \( (w_1, w_2, \dots, w_m ) \)
of expressions and returns \( G(Z; \sum_{i=1}^m w_i, \sum_{i=1}^m w_i^2, \dots ) \).

\begin{exa}
Letting \( G = \mathcal{D}_4 \) and setting \( w_1 = 1, w_2 = t \), we can obtain the generating function for \( 2 \)-colored necklaces of size \( 4 \):

\begin{snippet}
\codefont{In[2]:= dihedralGroupCycleIndex =  DihedralGroupIndex[4, x];}\\
\codefont{colorEnumerator = \{1, t\}; }

\codefont{In[3]:= OrbitInventory[dihedralGroupCycleIndex, x, colorEnumerator] }

\codefont{Out[3]:= 1 + t + 2 t\^{}2 + t\^{}3 + t\^{}4 }
\end{snippet}

It turns out that Combinatorica has a built-in rule for this type of result:

\begin{snippet}
\codefont{In[4]:= NecklacePolynomial[4, colorEnumerator, Dihedral]}

\codefont{Out[4]:= 1 + t + 2 t\^{}2 + t\^{}3 + t\^{}4 }
\end{snippet}
\end{exa}

\begin{exa}
As our last example of Combinatorica in action, we count the number of unlabeled graphs with \( 3 \) vertices.
To apply P\'olya's theorem,
we define the set of all \(2\)-subsets of \( [1..3] \), representing all \( \binom{3}{2} \) edges:

\begin{snippet}
\codefont{In[5]:= set = KSubsets[Range[3], 2]}

\codefont{Out[5]:= \{\{1,2\},\{1,3\},\{2,3\}\} }
\end{snippet}

The group of permutations to quotient out is the set containing each permutation of \codefont{set} obtainable by applying a permutation \( g \in S_3 \) to each \(2\)-set in \codefont{set}.
Combinatorica has the rule \codefont{KSubsetGroup} to create such a group:

\begin{snippet}
\codefont{In[6] := group = KSubsetGroup[SymmetricGroup[3], set];}
\end{snippet}

We stated \hyperref[exa:graphs]{above} that this group is isomorphic to \( S_3 \).
We can prove this with the Combinatorica rule \codefont{MultiplicationTable}, which takes a set and an operation and gives the group multiplication table of the group they form:

\begin{snippet}
\codefont{In[7] := MultiplicationTable[SymmetricGroup[3], Permute] // TableForm}

\codefont{Out[7]//Tableform=}
\[
\begin{array}{llllll}
 \codefont{1} & \codefont{2} & \codefont{3} & \codefont{4} & \codefont{5} & \codefont{6} \\
 \codefont{2} & \codefont{1} & \codefont{5} & \codefont{6} & \codefont{3} & \codefont{4} \\
 \codefont{3} & \codefont{4} & \codefont{1} & \codefont{2} & \codefont{6} & \codefont{5} \\
 \codefont{4} & \codefont{3} & \codefont{6} & \codefont{5} & \codefont{1} & \codefont{2} \\
 \codefont{5} & \codefont{6} & \codefont{2} & \codefont{1} & \codefont{4} & \codefont{3} \\
 \codefont{6} & \codefont{5} & \codefont{4} & \codefont{3} & \codefont{2} & \codefont{1}
\end{array}
\]

\codefont{In[8] := MultiplicationTable[group, Permute] // TableForm}

\codefont{Out[8]//Tableform=}
\[
\begin{array}{llllll}
 \codefont{1} & \codefont{2} & \codefont{3} & \codefont{4} & \codefont{5} & \codefont{6} \\
 \codefont{2} & \codefont{1} & \codefont{5} & \codefont{6} & \codefont{3} & \codefont{4} \\
 \codefont{3} & \codefont{4} & \codefont{1} & \codefont{2} & \codefont{6} & \codefont{5} \\
 \codefont{4} & \codefont{3} & \codefont{6} & \codefont{5} & \codefont{1} & \codefont{2} \\
 \codefont{5} & \codefont{6} & \codefont{2} & \codefont{1} & \codefont{4} & \codefont{3} \\
 \codefont{6} & \codefont{5} & \codefont{4} & \codefont{3} & \codefont{2} & \codefont{1}
\end{array}
\]
\end{snippet}

In order to apply P\'olya's enumeration theorem, we use \codefont{CycleIndex} to compute the cycle index of \codefont{group}:

\begin{snippet}
\codefont{In[9] := cycleIndex = CycleIndex[group, x]}

\codefont{Out[9] := x[1]\^{}3/6 + x[1] x[2]/2 + x[3]/3}
\end{snippet}

Now \codefont{OrbitInventory} can reproduce our result:

\begin{snippet}
\codefont{In[10] := OrbitInventory[cycleIndex,x,\{1,t\}] }

\codefont{Out[10] := 1 + t + t\^{}2 + t\^{}3}
\end{snippet}

\sloppy We did not actually have to do all this to find the total number, since Combinatorica includes the rules \codefont{NumberOfGraphs} and \codefont{ListGraphs} to count and list all nonisomorphic graphs with a given number of vertices.

\begin{snippet}
\codefont{In[11] := NumberOfGraphs[3]}

\codefont{Out[11] := 4}
\end{snippet}
\end{exa}

\subsection{GraphEnumeration}
Author: Doron Zeilberger
\\
Download: \url{http://www.math.rutgers.edu/~zeilberg/tokhniot/GraphEnumeration}
\\
Last Modified: July 20, 2010
\\
Website: \url{http://www.math.rutgers.edu/~zeilberg/mamarim/mamarimhtml/GE.html}
\\

This Maple package implements the P\'olya-theoretic methods of \cite{graphical} to solve various counting problems of the following type: How many nonisomorphic graphs are there with \(n\) vertices, \(m\) edges, and some property \(P\)?

The types of graphs that GraphEnumeration counts include
unlabeled connected graphs according to the number of edges,
unlabeled regular \( k \)-hyper-graphs according to edges,
unlabeled rooted trees,
unlabeled trees,
all (connected or general) unlabeled simple graphs with a given degree sequence, and
all (connected or general) unlabeled multi-graphs with a given degree sequence.

\subsection{PermGroup}

Author: Thomas Bayer
\\
Last modified: June 21, 2004
\\
Website: \url{http://www.risc.jku.at/research/combinat/software/PermGroup/}
\\

PermGroup is a Mathematica package for permutation groups, group actions and counting.
Documentation is scant,
but like Combinatorica, PermGroup has built-in rules for the most commonly used permutation groups,  methods for creating custom groups, and ways to use them to solve counting problems involving symmetry.
PermGroup implements the Schreier-Sims algorithm\footnote{
 \url{http://en.wikipedia.org/wiki/Schreier-Sims_algorithm}
} for computing group orders.
The examples below illustrate some of PermGroup's functionality, for more details consult the documentation and code.

\begin{exa}
To obtain the generating function for \(2\)-colored necklaces of length \(4\) with \(t\) marking the number of black beads,
we can  use the \codefont{DihedralGroupCIPoly} rule which returns the cycle index for the dihedral group.

\begin{snippet}
\codefont{In[1] := DihedralGroupCIPoly[4, x]}

\codefont{Out[1] := 1/4 (x[1]\^{}2 x[2] + x[2]\^{}2) + 1/8 (x[1]\^{}4 + x[2]\^{}2 + 2 x[4])}
\end{snippet}

Now we can replace \(x_n\) with \( 1 + t^n \):

\begin{snippet}
\codefont{In[2] := \% /. \{x[n\_] -> 1 + t\^{}n\}}

\codefont{Out[2] := 1/4 ((1 + t)\^{}2 (1 + t\^{}2) + (1 + t\^{}2)\^{}2) +  1/8 ((1 + t)\^{}4 + }\\
\codefont{(1 + t\^{}2)\^{}2 + 2 (1 + t\^{}4))}
\end{snippet}

Finally, we expand the polynomial and get something familar:

\begin{snippet}
\codefont{In[3] := Expand@\%}

\codefont{Out[3] := 1 + t + 2 t\^{}2 + t\^{}3 + t\^{}4}
\end{snippet}
\end{exa}

\begin{exa}
To obtain the generating function for unlabeled graphs with \( 3 \) vertices, we first define the set of all \(2\)-subsets of \( [1..3] \), representing all \( \binom{3}{2} \) edges:

\begin{snippet}
\codefont{In[4] := graphs3 = PermGroup\backtick PairSet[3]}

\codefont{Out[4] := \{\{1,2\},\{1,3\},\{2,3\}\}}
\end{snippet}

We let \codefont{s3} represent the symmetric group \(\mathcal{S}_3\) using the \codefont{SymmGroup} rule:

\begin{snippet}
\codefont{In[5] := s3 = PermGroup\backtick Generate[PermGroup\backtick SymmGroup[3]];}
\end{snippet}

Again, the group we really need is the induced group of \codefont{s3} on \codefont{graphs3}.
PermGroup can give this to us with the \codefont{TransformGroup} rule:

\begin{snippet}
\codefont{In[6] := g3 = TransformGroup[s3, graphs3, PairSetAction];}
\end{snippet}

Now we could compute the cycle index polynomial of \codefont{g3} and substitute in \( 1 + t^n \) as above, but
PermGroup offers the \codefont{PolyaEnumeration} rule to save some work.
The expression \codefont{PolyaEnumeration[group, x, list]} evaluates to the cycle index of \codefont{group} in \codefont{x[1], x[2], \dots} with \codefont{x[i]} replaced with \codefont{Total[list\^{}i]}.

\begin{snippet}
\codefont{In[7] := Expand[PolyaEnumeration[g3, x, \{1, t\}]] }

%
%
%
\codefont{Out[7] := 1 + t + t\^{}2 + t\^{}3}
\end{snippet}
\end{exa}

\section{Combinatorial species}
\label{sec:species}

\subsection{Mathematical background}

The theory of combinatorial species,
``the most fruitful unifying concept in enumerative combinatorics of this quarter-century''
according to Zeilberger,\footnote{
 \url{http://www.math.rutgers.edu/~zeilberg/khaver.html}
}
is a theoretical framework in the language of which a great diversity of families of combinatorial objects and their generating functions can be described.

While the theory of species and symbolic combinatorics (the subject of Section \ref{sec:symboliccombinatorics}) have a lot in common, one of the important differences is that in the theory of species, labeled and unlabeled objects are explicitly connected through P\'olya theory.
Martin Rubey, co-author of the Aldor-Combinat package described in the next section, said in 2008:
\begin{quote}
Combstruct and MuPAD-Combinat really had ``usability'' as
primary goal.  As one consequence, they do not implement species, but rather
``combinatorial classes'', that is, collections of objects with a size function.

The main drawback of that method is that you cannot treat labelled and
unlabelled objects uniformly, there is no such concept as an isomorphism type,
which, in my opinion, is the main strength of (ordinary) species.\footnote{
 \url{http://www.mail-archive.com/aldor-combinat-devel@lists.sourceforge.net/msg00503.html}
}
\end{quote}
Connections between symbolic combinatorics and P\'olya theory have been made \cite{ac}%
\footnote{
 \url{http://www.mathematik.uni-stuttgart.de/~riedelmo/papers/collier.pdf}
}, but they are not as central to the theory as they are in species.
The article \cite{Salvy2011} is ``self-contained and can be used as a dictionary between the theory of species and the symbolic method of Flajolet and Sedgewick''.
\\

We begin with the definition of a species:
\begin{defin}
A \emph{species (of structures)} \( F \) is a pair of functions,
\begin{enumerate}
 \item the first of which maps each finite set \(U\) to a finite set \(F[U]\), and
 \item the second of which maps each bijection \( \pi: U \rightarrow V \) to a function
\[ F[\pi]:F[U] \rightarrow F[V]. \]
\end{enumerate}
The functions \( F[\pi] \) must satisfy the following properties:
\begin{enumerate}
 \item for each pair of bijections \( \pi : U \rightarrow V \) and \( \tau : V \rightarrow W \),
\[ F[\tau \circ \pi] = F[\tau] \circ F[\pi] , \]
 \item and if \(\Id_U : U \rightarrow U\) is the identity function,
\[ F[\Id_U] = \Id_{F[U]}.\]
\end{enumerate}
\end{defin}

We note that if \( H \) is a group of permutations on \( U \), then the second function of \( F \) is a group homomorphism from \( H \) to \(\mathcal{S}(F[U]) \), i.e.\ an action of \( H \) on \(F[U]\).
We denote the image of this action by \( F[U; H] \).

\begin{defin}
An element \( s \in F[U] \) is called an \(F\)-\emph{structure} on \(U\), or alternatively a \emph{structure of species \(F\) on \(U\)}.
The function \( F[\pi] \) is called the \emph{transport} of \(F\)-structures along \( \pi \).
\end{defin}

If \( s \in F[U] \), we use the notation \( \pi \cdot s \) to denote \( F[\pi](s) \).
Note that \( |U| = |V| \) implies \( |F[U]| = |F[V]| \).

Whereas in symbolic combinatorics we had the classes \(\mathcal{E}\) and \(\mathcal{Z}\), in combinatorial species, the most elementary species are the \emph{empty set species} \(E\), where
\[ E[U] = \begin {cases}
\{\emptyset\}  & \text{if } |U|=0, \\
\emptyset & \text{otherwise},
\end {cases}
\]
and the \emph{singleton species} \(S\), where
\[ S[U] = \begin {cases}
\{u\}  & \text{if } U=\{u\}, \\
\emptyset & \text{otherwise}.
\end {cases}
\]

\begin{defin}
Consider two \(F\)-structures \(s_1 \in F[U]\) and \( s_2 \in F[V] \).
A bijection \( \pi: U \rightarrow V \) is called an \emph{isomorphism} of \(s_1\) to \(s_2\) iff \(s_2 = \pi \cdot s_1 \).
If there is such a bijection, we write \(s_1 \sim s_2 \) and we say these two structures have the same \emph{isomorphism type}.
The isomorphism types of a species are isomorphism classes and thus equivalence classes.
\end{defin}

An \(F\)-structure \(s \in F[U] \) on a set \( U \) can be referred to as a \emph{labeled} structure, whereas an isomorphism type of \( F \)-structures can be referred to \emph{unlabeled} structure.

Attentive readers will be detecting a whiff of P\'olya theory by now, which will become stronger as we define the generating functions associated with a species.

\begin{defin}
The \emph{exponential generating function} of a species \( F \) is the formal power series
\[ F(z) = \sum_{n \geq 0} f(n) \frac{z^n}{n!}, \]
where \( f(n) = |F[[1..n]]| \).
\end{defin}

\begin{defin}
\label{defin:t}
Denote by \( T(F_n) \) the \hyperlink{quotient}{quotient set} \(F[[1..n]]/ F[[1..n]; \mathcal{S}_n] \) of unlabeled \(F\)-structures on \( [1..n] \).
Then the \emph{isomorphism type generating function} of a species \(F\) is the formal power series
\[ \tilde{F}(z) = \sum_{n \geq 0} \tilde{f}(n) z^n, \]
where \(\tilde{f}(n) = |T(F_n)| \).
\end{defin}

As in Lemma \ref{lem:burnside}, if \( \pi \) is a permutation, let \(\fix(\pi)\) be the number of fixed points of \( \pi \).

\begin{defin}
Let \( \mathcal{S}_n \) denote the group of permutations of \( [1..n] \).
Then the \emph{cycle index generating function} of a species \( F \) is the formal power series
\[ Z_F(\boldsymbol{z}) = \sum_{n \geq 0}\frac{1}{n!}  \sum_{\pi \in S_n} \fix(F[\pi]) \boldsymbol{z}^{\boldsymbol{c}(\pi)} . \]
\end{defin}

\begin{thm}
For any species \( F \) we have
\begin{enumerate}
 \item \(F(z) = Z_F(z,0,0,\dots) \)
 \item \( \tilde{F}(z) = Z_F(z, z^2, z^3, \dots ) \).
\end{enumerate}
\end{thm}

\begin{proof}We prove each part separately.
\begin{enumerate}

 \item We have
\begin{align*}
 Z_F(z,0,0,\dots) &= \sum_{n \geq 0}\frac{1}{n!}  \sum_{\pi \in S_n} \fix(F[\pi]) z^{c_1(\pi)}  0^{c_2(\pi)} 0^{c_3(\pi)} \cdots  \\
&= \sum_{n \geq 0} \frac{1}{n!} \fix(F[\Id_{[1..n]}]) z^n \\
&= \sum_{n \geq 0}  f(n) \frac{z^n}{n!}.
\end{align*}

 \item Now,
\begin{align*}
 Z_F(z, z^2, z^3, \dots ) &= \sum_{n \geq 0}\frac{1}{n!}  \sum_{\pi \in S_n} \fix(F[\pi]) z^{c_1(\pi)}  z^{2 c_2(\pi)} z^{3 c_3(\pi)} \cdots  \\
&= \sum_{n \geq 0}\frac{1}{n!}  \sum_{\pi \in S_n} \fix(F[\pi]) z^n
\end{align*}

Let \( G = F[[1..n], \mathcal{S}_n] \) be the image of the action of \( \mathcal{S}_n \) on \(F[[1..n]] \) given by the second function of \( F \).
Then we have
\begin{align*}
\sum_{n \geq 0}\frac{1}{n!}  \sum_{\pi \in S_n} \fix(F[\pi]) z^n &= \sum_{n \geq 0}\frac{1}{n!}  \frac{n!}{|G|} \sum_{\fix(F[\pi]) \in G} \fix(F[\pi]) z^n \\
&= \sum_{n \geq 0}\frac{1}{|G|} \sum_{\tau \in G} \fix(\tau) z^n \\
&= \sum_{n \geq 0} \tilde{f}(n) z^n,
\end{align*}
where the last equality follows from Lemma \ref{lem:burnside}. \qedhere
\end{enumerate}
\end{proof}

\begin{defin}
\label{defin:speciesequality}
Let \(F\) and \(G\) be two species.
An \emph{isomorphism} of \(F\) to \(G\) is a family of bijections \( b_U : F[U] \rightarrow G[U] \) which satisfies the following condition:
for any bijection \( \pi : U \rightarrow V \) between two finite sets and any \( F \)-structure \(s \in F[U] \), we must have \( \pi \cdot b_U(s) = b_V(\pi \cdot s) \).
The two species are then said to be \emph{isomorphic} and we write \(F = G \).
\end{defin}

As with combinatorial classes from symbolic combinatorics, there is a set of operations on species which have corresponding operations on their generating functions (all three types).
These operations include addition, multiplication, composition, differentiation,
pointing, cartesian product, and functorial composition, just as in symbolic combinatorics.
For more information on these, and all other parts of the theory of species, see the 457-page book devoted to the subject \cite{speciesbook}.
The species analogs of combinatorial specifications are systems of equations, which can be created using the operators just mentioned and the notion of equality from Definition \ref{defin:speciesequality}.

\subsubsection{Weighted species}

The species equivalent of multiple parameters on combinatorial classes is weighted species, which we briefly discuss.

A \emph{weighting function} \( w: F[U] \rightarrow R[t] \) maps objects in \( F[U] \) to monomials in \( t \) over a ring \( R \subseteq \mathbb{C} \),
and an \( R[t] \)-\emph{weighted set} is a pair \( (S,w) \), where \( w: S \rightarrow R[t] \) is a weighting function on \( S \).
For such a pair \( (S,w) \), let \( |S|_w = \sum_{s \in S} w(s) \).


\begin{defin}
Let \( R \subseteq \mathbb{C} \) be a ring.
An \(R[t]\)-\emph{weighted species} \( F \) is a pair of functions,
\begin{enumerate}
 \item the first of which maps each finite set \(U\) to an \(R[t]\)-weighted set \( (F[U], w_U )\) such that \(|F[U]|_{w_U}\) converges, and
 \item the second of which maps each bijection \( \pi: U \rightarrow V \) to a function
\[ F[\pi]:(F[U], w_U) \rightarrow ( F[V], w_V ), \]
that preserves weights.
\end{enumerate}
The functions \( F[\pi] \) must satisfy the following properties:
\begin{enumerate}
 \item for each pair of bijections \( \pi : U \rightarrow V \) and \( \tau : V \rightarrow W \),
\[ F[\tau \circ \pi] = F[\tau] \circ F[\pi] , \]
 \item and if \(\Id_U:U \rightarrow U\) is the identity function,
\[ F[\Id_U] = \Id_{F[U]}.\]
\end{enumerate}
\end{defin}

\begin{defin}
Let \( F \) be a weighted species,
with weight functions \(w_n : F[[1..n]] \rightarrow R[t] \), for \( n \geq 0 \).
Then the exponential generating function of \( F \) is
\[ F_w(z) = \sum_{n \geq 0} |F[[1..n]]|_{w_n} \frac{z^n}{n!}; \]

\noindent the isomorphsim type generating function of \(F\) is
\[ \tilde{F}_w(z) = \sum_{n \geq 0} |T(F_n)|_w z^n, \]

\noindent where \( T(F_n) \) is defined as in Definition \ref{defin:t};
and the cycle index generating function of \( F \) is
\[ Z_{F_w}( \boldsymbol{z} ) = \sum_{n \geq 0} \frac{1}{n!} \left( \sum_{ \pi \in \mathcal{S}_n} |\Fix(F[\pi])|_{w_n} \boldsymbol{z}^{\boldsymbol{c}(\pi)} \right), \]
where, if \( \pi \in \mathcal{S}_n \),  \( \Fix(F[\pi]) \) is the set of \( F \)-structures on \( [1..n] \) fixed by \( F[\pi] \).
\end{defin}

Operations on weighted species can be defined similarly to the unweighted case.
For more information, see \cite{speciesbook}.

\begin{exa}
Let \( L_{ \geq 1} \) be the species of linear orders resricted to sets of at least one element;
let \( S \) be the singleton species with weight \( 1 \); and let \( W \) be the singleton species with weight \( q \).
Then the species \( T \) of ordered trees with number of internal nodes marked by \( q \) satisfies
\[ T=S + W \cdot (L_{\geq 1} \circ T) \]
where \( + \) is addition (the species analog of sum), \( \cdot \) is multiplication (the analog of product), and \( \circ \) is composition (the analog of substitution).\\
\end{exa}

Two notable packages for species, Darwin \cite{darwin1,darwin2}, and Devmol \cite{devmol} have been described in the literature, but no longer exist.

\subsection{Aldor-Combinat}
Authors: Ralf Hemmecke and Martin Rubey
\\
Website: \url{https://portal.risc.jku.at/Members/hemmecke/aldor/combinat}
\\

A project which started in 2006, Aldor-Combinat is an unfinished package primarily for working with species,
which is based on MuPAD-Combinat's implementation of symbolic combinatorics.
It is written in the Aldor language to be used with the Axiom\footnote{
 \url{http://axiom-developer.org/}
} computer algebra system.
The only included documentation contains many of the implementation details of the package, but
it is designed to be read by Aldor developers.\footnote{
 \url{http://www.risc.jku.at/people/hemmecke/AldorCombinat/combinat.html}
}

Like Sage uses for basic combinatorial objects, Aldor-Combinat uses a category-theoretic model to organize its functionality.
All species in Aldor-Combinat are objects in the  \codefont{CombinatorialSpecies} category and they must implement methods returning
exponential generating series,
isomorphism type generating series,
cycle index generating series,
a listing of structures,
and a listing of isomorphism types.

Aldor-Combinat has the following species predefined:
empty set,
set,
singleton,
linear order,
cycle,
permutation, and
set partition (whose isomorphism types are integer partitions).

\begin{exa}

We can compute the initial coefficients of the exponential generating function of the singleton species:
\begin{snippet}

\codefont{L == Integer;}

\codefont{E == EmptySetSpecies L;}

\codefont{gse: OrdinaryGeneratingSeries := generatingSeries \$ E;}

\codefont{import from Integer;}

\codefont{le: List Integer := [coefficient(gse, n) for n in 0..3];}

\end{snippet}
\noindent This assigns \codefont{[1, 0, 0, 0]} to \codefont{le}.

We can compute something similar for the singleton species:
\begin{snippet}

\codefont{M == Integer;}

\codefont{S == SingletonSpecies M;}

\codefont{gss: ExponentialGeneratingSeries := generatingSeries \$ S;}

\codefont{import from Integer;}

\codefont{ls: List Integer := [coefficient(gss, n) for n in 0..3];}
\end{snippet}

\noindent This assigns \codefont{[0, 1, 0, 0]} to \codefont{ls}.
\end{exa}

Let \( F \) be a species.
Define, for each \( n \geq 0 \), the \emph{species \( F \) restricted to \( n \)} by
\[ F_n[U] = \begin{cases}
 F[U] & \text{if } |U|=n, \\
 \emptyset & \text{otherwise},
\end{cases}
\]
for all finite sets \( U \).
Aldor-Combinat allows species to be restricted with \codefont{RestrictedSpecies},
which is useful for defining other species.

Species may be defined implicitly or explicitly with the following
operations:
addition,
multiplication,
composition, and
functorial composition.

\subsection{Sage: combinatorial species}
\label{sec:sagespecies}

Download: \url{http://www.sagemath.org/download.html}
\\
Website: \url{http://sagemath.org/doc/reference/combinat/species.html}
\\

This section covers Sage's species functionality.
For information on Sage's basic combinatorial objects functionality and an overview Sage as a CAS, see Section \ref{sec:alas}.
\\

Sage's capabilities for working with combinatorial species,
which began development around 2008,
are based on those in Aldor-Combinat (covered in the previous section).
A project roadmap describes future plans for the project.\footnote{
 \url{http://trac.sagemath.org/sage_trac/ticket/10662}
}

Sage offers built-in classes representing the
cycle,
partition,
permutation,
linear-order,
set, and
subset
species.
Also, the \codefont{CharacteristicSpecies(n)} method returns the \emph{characteristic species on \codefont{n}}, which
yields
one object on sets of size \codefont{n}, and no objects on any other set (i.e.\ the restriction to \codefont{n} of the set species \(X\), defined \(X[U]=U\)).

\begin{exa}
We assign an empty set species to \codefont{E}:

\begin{snippet}
\codefont{sage: E = species.EmptySetSpecies()}
\end{snippet}

We list all structures of \codefont{X} on the empty set and the set \( \{1,2\} \):

\begin{snippet}
\codefont{sage: E.structures([]).list()}

\codefont{[\{\}]}

\codefont{sage: E.structures([1,2]).list()}

\codefont{[]}
\end{snippet}

We find the first four coefficients of the exponential generating function of \codefont{X}:

\begin{snippet}
\codefont{sage: E.generating\_series().coefficients(4)}

\codefont{[1, 0, 0, 0]}
\end{snippet}

Since the empty set species is isomorphic to the characteristic species on \( 0 \), we get identical output using \codefont{CharacteristicSpecies(0)}:

\begin{snippet}
\codefont{sage: C0 = species.CharacteristicSpecies(0)}

\codefont{sage: C0.structures([]).list()}

\codefont{[\{\}]}

\codefont{sage: C0.structures([1,2]).list()}

\codefont{[]}

\codefont{sage: C0.generating\_series().coefficients(4)}

\codefont{[1, 0, 0, 0]}
\end{snippet}

The characteristic species on \( 1 \) is isomorphic to the singleton species:

\begin{snippet}
\codefont{sage: C1 = species.CharacteristicSpecies(1)}

\codefont{sage: C1.structures([1]).list()}

\codefont{[1]}

\codefont{sage: C1.structures([1,2]).list()}

\codefont{[]}

\codefont{sage: C1.generating\_series().coefficients(4)}

\codefont{[0, 1, 0, 0]}
\end{snippet}
\end{exa}

The operations on species of
addition,
multiplication,
composition, and
functorial composition are supported.

Species can be defined explicitly or implicitly (\codefont{define}).
When species are created, a weight can be specified, as well as a restriction on the size of the set.

Each species object \codefont{B} implements a number of useful methods:
\codefont{B.structures(set)} is the set of \codefont{B}-structures on the set \codefont{set},
\codefont{B.isotypes(set)} is the set of equivalence classes of \codefont{B} structures (isomorphism types) on the set \codefont{set},
\codefont{cycle\_index\_series()} is the cycle index series,
\codefont{generating\_series()} is the exponential series, and
\codefont{isotype\_generating\_series()} is the isomorphism type series.

\begin{exa}
We wish to count unlabeled ordered trees by total number of nodes and number of internal nodes.
To achieve this, we begin by assigning a weight of \(1\) to the leaves and \(q\) to internal nodes, which are each singleton species:

\begin{snippet}
\codefont{sage: q = QQ['q'].gen()}

\codefont{sage: leaf = species.SingletonSpecies()}

\codefont{sage: internal\_node = species.SingletonSpecies(weight=q)}
\end{snippet}

Now we define a species \codefont{T} representing the trees, defined as \codefont{leaf + internal\_node*L(T)}, where \codefont{L} is the species of linear orders restricted to sets of size \(1\) or greater:

\begin{snippet}
\codefont{sage: L = species.LinearOrderSpecies(min=1)}

\codefont{sage: T = species.CombinatorialSpecies()}

\codefont{sage: T.define(leaf + internal\_node*L(T))}
\end{snippet}

All that remains, since the trees are unlabeled, is to compute the coefficients of the isomorphism type generating function:

\begin{snippet}
\codefont{sage: T.isotype\_generating\_series().coefficients(6)}

\codefont{[0, 1, q, q\^{}2 + q, q\^{}3 + 3*q\^{}2 + q, q\^{}4 + 6*q\^{}3 + 6*q\^{}2 + q]}
\end{snippet}
\end{exa}

Further examples may be found in a Sage species demo by Mike Hansen, one of the developers.\footnote{
 \url{http://sage.math.washington.edu/home/mhansen/CombinatorialSpeciesDemo.html}
}

\section{Asymptotics}
\label{sec:asymptotics}
\subsection{Mathematical background}
\subsubsection{Asymptotic scales and series}

\begin{defin}
Let \( P \) be an open set and let \( L \) be a limit point of \( P \).
Let \(S_P\) be a set of functions from \( P \) to \( \mathbb{C} \).
Then the set \(S_P \) is an \emph{asymptotic scale} at \(L \) iff for every pair of different functions \( \phi_1, \phi_2 \in S_P \), the limit of \( \phi_1(x)/\phi_2(x) \) as \(x \rightarrow L\) is either \(0\) or \(+\infty\).
\end{defin}

\begin{exa}
\label{exa:polys}
If \( P = \mathbb{Z}_{\geq0} \cup \{ \infty \} \), \(L = \infty \), then an asymptotic scale is \( S_P= \{n \mapsto 1, n \mapsto n, n \mapsto n^2,\dots \} \).
\end{exa}

The following definition, which comes from \cite{gdevexamples} (as does the previous definition), generalizes Definition \ref{def:asymptoticseries} for terminating series:
\begin{defin}
Let \(S_P\) be an asymptotic scale at \(L\) and let \(f\) be a complex valued function on \( P \).
Then \(f\) is said to admit a \emph{terminating asymptotic series} at \(L\) iff there exists a sequence \( (a_1, \dots, a_n) \) of complex numbers and a sequence \( (\phi_1, \dots, \phi_n) \) of elements of \(S_P\) such that
\[ \forall { i \in \{1,\dots, n \}, \phi \in S_P}: \qquad \phi_i(x) = o(\phi(x)) \implies f(x) - \sum_{j=1}^{i} a_j \phi_j(x) = o(\phi(x)) \text{ as } x \rightarrow L.\]
\end{defin}

\begin{exa}
With \(P, L,\) and \(S_L\) defined as in Example \ref{exa:polys}, let \(f_k(n) = [z^n] \frac{z^k}{(1-z)^k}\) be the number of integer compositions of \(n\) with \(k\) parts.
Then a terminating asymptotic series for \(f_k(n)\) at \(L\) consists of the complex numbers \( \left( \frac{1}{(k-1)!}, -\frac{k(k-1)}{2(k-1)!} \right) \) and functions \( ( n \mapsto n^{k-1}, n \mapsto n^{k-2}) \) since
\[ f_k(n) = \frac{n^{k-1}}{(k-1)!} - \frac{k(k-1) n^{k-2}}{2 (k-1)!}  + o(n^{k-2}). \]
\end{exa}

Various CASes offer support for generating and working with asymptotic series expansions, including Maple, Mathematica, Maxima,
and MATLAB's Symbolic Math Toolbox (the current product containing the code of MuPAD)\footnote{
 \url{http://www.mathworks.com/help/toolbox/mupad/stdlib/asympt.html}
};
in this section we include a couple of implementations of special interest.

\subsubsection{Singularity analysis}
Singularity analysis is a method due to Flajolet and Odlyzko which shows how to compute asymptotic expressions for the coefficients of some classes of generating functions \cite{singularityanalysis}.
In order to apply the theorem, we forget about \emph{formal} power series and treat generating functions as elements of \( C^\omega(\mathbb{C}; 0) \) instead of \( \mathbb{Q}[[z]] \).
In its simplest form, the method can be expressed as the following theorem:
\begin{thm}[Singularity analysis]
Fix the range of  \( \Arg \) to be \( [-\pi, \pi) \).
Let \( F \in C^\omega(\mathbb{C}; 0) \) be a function analytic in a domain
\[ D = \left\{z : |z| \leq s_1, |\Arg(z-s)| > \frac{\pi}{2} - \eta \right\}, \]
where \( s, s_1 > s, \) and $\eta$ are three positive real numbers.  Assume that, with \( \sigma(u) = u^\alpha \log^\beta u \) and \( \alpha \notin \{0, -1, -2, \dots \} \), we have
\[ F(z) \sim \sigma \left( \frac{1}{1-z/s} \right) \qquad \text{ as } z \rightarrow s \text{ in } D. \]
Then, as \( n \rightarrow \infty \), the Maclaurin coefficients of \( F \) satisfy
\[ [z^n]F(z) \sim s^{-n} \frac{\sigma(n)}{n \Gamma(\alpha)}. \]
\end{thm}
Generalizations are discussed at length in \cite{ac}.

\subsubsection{Saddle-point asymptotics}

The saddle-point method is another way of computing asympotic expressions for the coefficients of generating functions.
In \cite{hayman}, W.K.\ Hayman defined a set of functions called the \emph{H-admissible functions} for which the following theorem holds:
\begin{thm}
If \( F \) is a function defined around the origin by \( F(z) = \sum_{n \geq 0} f(n) z^n \) and \( F \) is H-admissible, then
\[ f(n) \sim \frac{F(r(n))}{r(n)^n \sqrt{2 \pi b(r(n))}}, \text{ as } n \rightarrow \infty \]
where \(r(n) \) is the smallest positive root of \(r(n) F'(r(n)) / F(r(n)) = n\) and \(b(r) = r D_r (r F'(r)) \).
\end{thm}
Useful information on which functions are H-admissible is given by the following theorem (also from \cite{hayman}):
\begin{thm}
The following are three properties of H-admissible functions.
\begin{enumerate}
\item Let \( \alpha \) and \(\beta_1\) be positive real numbers, and let \( \beta_2 \) and \( \beta_3 \) be real numbers.
Then \( F \) defined by
\[ F(z) = \exp \left( \beta_1 (1-z)^{-\alpha} \left( \frac{1}{z} \log \frac{1}{1-z} \right)^{\beta_2} \left( \frac{2}{z} \log \left(\frac{1}{z} \log\frac{1}{1-z} \right) \right)^{\beta_3} \right) \]
is H-admissible.
\item If \( F \) and \( G \) are H-admissible, and \(P\) is a polynomial function with real coefficients and positive leading coefficient, then \( \exp(F), F+G, F+P, P(F), \) and \( P \cdot F\) are H-admissible.
\item If \(P \) is a polynomial function such that \( P(z) \) cannot be written as \(P(Q(z^k)) \) for any polynomial function \( Q \) and \( k >1 \), then \( \exp(P) \) is H-admissible.
\end{enumerate}
\end{thm}
The saddle-point method are discussed in further detail in \cite{ac}.

\subsubsection{Power series coefficients and the radius of convergence}
\begin{notation}[from \cite{ac}]
Given a sequence \( (a_0, a_1, a_2, \dots ) \) of rational numbers and \( \mathfrak{r} \in [0,\infty] \), we write
\[ a_n \bowtie \mathfrak{r}^n \]
iff
\[ \limsup_{n \geq 0} |a_n|^{1/n} = \mathfrak{r}, \]
where \(1/\infty = 0, 1/0 = \infty \).
\end{notation}

\begin{pro}
Given a sequence \( (a_0, a_1, a_2, \dots ) \) of rational numbers and \( \mathfrak{r} \in [0,\infty] \), we have
\[ a_n \bowtie \mathfrak{r} \text{ iff } a_n \sim \mathfrak{r}^{n}c(n), \]
where \( \limsup |c(n)|^{1/n} = 1 \).
In other words \( f_n \bowtie \mathfrak{r}\) iff \( 1/\mathfrak{r} \) is the radius of convergence of \( \sum_{n \geq 0} a_n z^n \).
\end{pro}

\begin{defin}

Given a function \( F \)
with Maclaurin series expansion \( F(z) = \sum_{n \geq 0} f(n) z^n \) having radius of convergence \( r \), with a finite set of singularities at distance \( r \) from the origin,
the \emph{dominant directions} of \( F \) is the set \( \Theta \) where
\[ \theta \in \Theta \iff \theta \in [-\pi, \pi) \text{ and } F \text{ is singular at } r e^{i \theta}. \]
\end{defin}


\subsection{Gdev}
\label{sec:gdev}

Author: Bruny Salvy
\\
Download: \url{http://algo.inria.fr/libraries/\#down}
\\
Webpage: \url{http://algo.inria.fr/libraries/\#gdev}
\\
Last modified: March 2003\footnote{
 \url{http://algo.inria.fr/salvy/index.html}
}
\\

Gdev is a package for computing asymptotic expressions for functions and generating function coefficients,
which,
except for the \codefont{equivalent} function (described below),
has been superseeded by the MultiSeries package (see Section \ref{sec:multiseries}).\footnote{
 \url{http://algo.inria.fr/salvy/index.html}
}.
For more
information about the capabilities of gdev
see an article written by Bruno Salvy \cite{gdevexamples}.

Many packages provide routines for computing Taylor expansions, but the capabilities of gdev's function \codefont{gdev} go well beyond that particular asymptotic scale.
The function \codefont{gdev}
takes a function given by an explicit expression
and,
automatically choosing the asymptotic scale to use,
finds an asymptotic expansion around a point, which ends with a big O term.

The function \codefont{equivalent}, originally in \hyperref[sec:ananas]{ANANAS} from LUO, extracts an asymptotic expression for the coefficients of an explictly given generating function using singularity analysis or the saddle-point method.
See the next section for more information on the underlying algorithm.

\subsection{Lambda-Upsilon-Omega (LUO): asymptotics}
\label{sec:ananas}

Authors: Bruno Salvy and Paul Zimmermann
\\
Download: \url{http://www.loria.fr/~zimmerma/software/luoV2.1.tar.gz}
\\
Last modified: May 1995
\\
Website: \url{http://algo.inria.fr/libraries/libraries.html\#luo}
\\

This section covers ANANAS, the component of LUO for extracting asymptotic information from generating functions.
To find out about LUO's other component, ALAS, and how ANANAS fits in to the LUO system, see Section \ref{sec:alas}.
\\

Once ALAS produces a system of equations for \(A^{<1>}(z), \dots, A^{<r>}(z) \), an attempt to solve it is made.  If an explicit expression is obtained, it is passed to ANANAS, written by
Bruno Salvy.

Not all explicit expressions can be handled by ANANAS, however --- in particular, not all expressions produced by ALAS.

\begin{defin}
The set \( \mathscr{E} \) of functions is the set of generating functions corresponding to labeled classes defined by well-defined explicit specifications.  That is, the functions containing the monomial functions \(z \mapsto 1\) and \( z \mapsto z \), and closed under the operations \( \{+,\times, Q,L,E \} \), where
\[ Q(F) = \frac{1}{1-F}, \qquad L(F) = \log\frac{1}{1-F}, \qquad E(F) = \exp(F). \]
\end{defin}
ANANAS can handle some functions not in \( \mathscr{E} \),
but for the rest of this section, we only consider ANANAS's behavior on functions in \( \mathscr{E} \),
since the documentation focuses on this case.

The set \( \mathscr{E} \) is subdivided into three disjoint sets
\[ \mathscr{E} =  \mathscr{E}_{AL} \hspace{3pt} \dot{\cup} \hspace{3pt} \mathscr{E}_{\text{entire}} \hspace{3pt} \dot\cup \hspace{3pt} \mathscr{E}_{\text{other}}, \]
where \( \mathscr{E}_{\text{entire}} \) is the set of entire functions and \( \mathscr{E}_{\text{AL}} \) is the set of functions with \emph{algebraic-logarithmic} (AL) growth,
meaning that \(F \in \mathscr{E}_{\text{AL}} \) iff \( F \) has radius of convergence \(0 < r < \infty \) and there exist \(\alpha \in \mathbb{R}, k \in \mathbb{Z}_{\geq 0} \) such that,
\[ F(z) \sim \frac{1}{(1  -z/r)^\alpha} \log^k \frac{1}{1 -z/{r}} \qquad \text{ as } z \rightarrow {r^-} . \]

ANANAS determines which subset of \( \mathscr{E} \) a function is a member of and proceeds accordingly.
For the details of the steps involved in this choice, see \cite{assistant}.
ANANAS is able to handle some functions from \(\mathscr{E}_{\text{entire}}\) and \(\mathscr{E}_{\text{other}}\) using the saddle-point method, but the coverage is not as complete as for \( \mathscr{E}_{AL} \), where singularity analysis is very effective.
We briefly describe the algorithm used for the class \(\mathscr{E}_{\text{AL}}\);
for information on what ANANAS does with \(\mathscr{E}_{\text{entire}}\) and \(\mathscr{E}_{\text{other}}\), see \cite{assistant}.

Given a function in \(\mathscr{E}_{\text{AL}}\), ANANAS performs the algorithm Equivalent:
\begin{center}
  \captionof{algorithm}{Equivalent}
\begin{algorithmic}
\Require \( F \) is an expression corresponding to a function in \( \mathscr{E}_{\text{AL}} \)
\State{1. Determine the radius of convergence \( r \) of \(F\)}

\State{2. Find the dominant directions \( \Theta \) of \( F \),}
\If {\( F \) has a single dominant direction}
        \State 3.1 Obtain an algebraic-logarithmic asymptotic expansion of \( F \) around \( r \)
        \State 3.2 Apply singularity analysis
\Else
        \State 4.1 Apply a singularity-analysis style result for functions with more than one dominant singularity
\EndIf
\end{algorithmic}
\end{center}
Step \(1\) turns out to be simple; it takes the form of the Radius algorithm:

\begin{center}
  \captionof{algorithm}{Radius}
\begin{algorithmic}
\Require \( F \) is an expression corresponding to a function in \( \mathscr{E}_{\text{AL}} \)
\If{ \( F \) is a polynomial}
 \State {then return \( \infty \)}

\ElsIf{\( F \) matches \( \exp(G) \)}
 \State{return the result of Radius on \( G \)}

\ElsIf{ \( F \) is \( Q(G) \) or \( L(G) \) }
 \State{ return the smallest real positive root of \(g(x)=1\)}

\ElsIf{ \( F \) matches \( G_1 + G_2\) or \( G_2 \cdot G_2 \)}
 \State{return the minimum of Radius on \(G_1\) and \(G_2\) }

\EndIf
\end{algorithmic}
\end{center}
For more details on the rest of Equivalent, see \cite{assistant}.

\subsection{MultiSeries}
\label{sec:multiseries}

Authors: Bruno Salvy and Maplesoft
\\
Download: \url{http://algo.inria.fr/libraries/\#down}
\\
Website: \\
\indent \url{http://www.maplesoft.com/support/help/Maple/view.aspx?path=MultiSeries}
\\

Originally developed by Bruno Salvy at the INRIA Algorithms Group and now built into Maple releases,
the MultiSeries Maple package provides functionality for computing asymptotic series expansions that superseeds that provided in gdev, the subject of Section \ref{sec:gdev}.

MultiSeries contains the function \texttt{multiseries}, the sucessor to gdev's \texttt{gdev}, and four special cases of \texttt{multiseries} that cover the same functionality of core Maple functions: \texttt{asympt}, \texttt{series}, \texttt{limit}, and \texttt{taylor}.
According to the Maple documentation, ``the simplest use of the package is by overriding the standard \texttt{asympt}, \texttt{series}, \texttt{limit} using \texttt{with}. The corresponding MultiSeries functions are often more powerful than the default ones, but require more computational time.''
(Presumably \texttt{taylor} can be used in the same way.)

The \texttt{multiseries} function can take as arguments an expression \texttt{expr}, the limit point \texttt{a}, the variable \texttt{x} in the expression tending to \texttt{a}, and the truncation order \texttt{n}.
It returns, like \texttt{gdev}, an asymptotic expansion of \texttt{expr} as \texttt{x} tends to \texttt{a} ending in a big O term, in terms of an asymptotic scale automatically chosen.

It is also possible to specify some extra options such as choosing the path of approach, and there is support for choosing the asymptotic scale used.\footnote{
 \url{http://www.maplesoft.com/support/help/Maple/view.aspx?path=MultiSeries\%2fmultiseries}
}

\section{Conclusions}
\label{sec:conc}
A number of kinds of conclusions can be made from this document.\\

First, recall that we
\hyperlink{started}{started this document}
with the claim that ``[it is no longer] necessary to use error-prone pen and paper methods to perform an ever-growing set of mathematical procedures''.
Indeed, the packages in this document cover a lot of ground, and much of the progress is recent, but there is massive potential for more work both with mathematical algorithms that have been implemented and those that have not.

Note that there are few packages available for any one CAS/programming language.
The average mathematician cannot be expected to be familiar with more than one system, if any.
There is therefore a need for projects to be ported to systems without those capabilities, preferably the big three CASs: Maple, Mathematica and Sage.\\

%
%

Next, future package-writers can note that not all of the packages in this
document come with sufficient documentation for all uses.
To be maximally helpful, two forms of documentation should be available: one for mere users, and one for developers or other people interested in the implementation details.
Users want to know what the package can do and how to do those things.
Developers wishing to port or extend the software want to know the algorithms and code used.\\

Lastly, as casual observations, we can say that learning how to use a package can be a good way to learn a mathematical concept or increase one's understanding, and also that
these packages are for the most part not huge collaborative efforts ---
often there is only one author, sometimes two.\\

Writing symbolic packages is not trivial, since it requires skills from multiple disciplines, namely mathematics, computer science, and software engineering.
However, it is clear that the future of the field is bright, and
the global project to replace pen and paper, and ultimately more and more of
the problem solving, theorem proving and question answering now done only by
human mathematicians will continue to progress and advance significantly in the
years to come.

\section*{Acknowledgements}
\addcontentsline{toc}{section}{Acknowledgements}
The author thanks Daniel Panario and Brett Stevens for their help with
this document.

\addcontentsline{toc}{section}{References}

\bibliographystyle{plain}
\bibliography{software}

\section*{Appendix}
\addcontentsline{toc}{section}{Appendix}

\subsection*{Maple 16}
\addcontentsline{toc}{subsection}{Maple 16}
Maple is a modern commercial CAS whose development began in 1980 by the Symbolic Computation Group at the University of Waterloo, and, since 1988, has been developed and sold by Waterloo Maple Inc., a.k.a. Maplesoft.
The most recent version as of time of writing is Maple 16, released in 2012.

\subsection*{Mathematica 8}
\addcontentsline{toc}{subsection}{Mathematica 8}
Mathematica, another modern commercial CAS, is developed by Wolfram Research.
It began as a project called Symbolic Manipulation Program created by Stephen Wolfram around 1979 at Caltech,
and then after Wolfram Research was founded in 1987,
was renamed Mathematica.
Version 1.0 of Mathematica was released in 1988;
the most recent release was Mathematica 8, in 2010.

Mathematica is a Turing complete \emph{term rewriting system}, which means that all Mathematica programs compute by manipulating expressions (terms).
When an expression is \emph{evaluated} by Mathematica,
it searches through
a sequence of stored pattern-replacement rules (in which form Mathematica's mathematical routines exist) looking for a match between a pattern and a subexpression of the expression.
When a rule's pattern matches a subexpression, that subexpression is replaced according to the rule, and the new expression is evaluated.
For more on term rewriting systems, see \cite{term}, and for more on Mathematica programming, see its documentation and \cite{trott}.

\end{document}